%% file: main-FV.tex
\documentclass[11pt]{article}

\def\fullversion{1}
\usepackage{fullpage}
\input{envs}

\title{Constant Approximation for $k$-Median and $k$-Means with Outliers via Iterative Rounding}
\author{Ravishankar Krishnaswamy\thanks{Microsoft Research India. Email: \texttt{ravishan@cs.cmu.edu}} \and Shi Li\thanks{Department of Computer Science and Engineering, University at Buffalo. Email: \texttt{shil@buffalo.edu}} \and Sai Sandeep\thanks{Microsoft Research India. Email: \texttt{saisandeep192@gmail.com}}}

\begin{document}

\maketitle

\input{abstract}

\input{intro.tex}

\input{prelim}

\input{integral.tex}

\input{framework.tex}

\input{fractional.tex}

\input{matroid.tex}

\input{knapsack.tex}

\paragraph{Acknowledgements.} RK would like to thank Deeparnab Chakrabarty, Shivam Garg and Prachi Goyal for several useful discussions that led to this work. The work of SL is supported by NSF grants CCF-1566356 and CCF-1717134.

\bibliography{paper}



\end{document}

%% file: envs.tex
\usepackage{microtype}
\usepackage{hyperref}

\usepackage{cleveref,xspace}
\usepackage{bbm}
\usepackage{times}
\usepackage{graphicx}
\usepackage{amsmath,amssymb,amsthm,mathtools}
\usepackage{paralist}
\usepackage{bm}
\usepackage{xspace}
\usepackage{url}
\usepackage{prettyref}
\usepackage{boxedminipage}
\usepackage{wrapfig}
\usepackage{ifthen}
\usepackage{color}
\usepackage{xcolor}
\usepackage{framed}
\usepackage{algorithm}
\usepackage{algpseudocode}
\usepackage{enumitem}
\algtext*{EndWhile}
\algtext*{EndIf}
\algtext*{EndFor}

\usepackage{thmtools}

\allowdisplaybreaks

\newtheorem{theorem}{Theorem}[section]

\newtheorem{lemma}[theorem]{Lemma}
\newtheorem{claim}[theorem]{Claim}

\newtheorem{definition}[theorem]{Definition}

\usepackage{cleveref}

\newcommand{\kmwo}{{\sf RkMed}\xspace}
\newcommand{\kmeanswo}{{\sf RkMeans}\xspace}
\newcommand{\matmed}{{\sf MatMed}\xspace}
\newcommand{\knapmed}{{\sf KnapMed}\xspace}

\newcommand{\sse}{\ensuremath{\subseteq}\xspace}

\newcommand{\R}{\mathbb R}

\newcommand{\eps}{\varepsilon}

\newcommand{\calI}{{\mathcal{I}}}

\newcommand{\cfull}{\ensuremath{{C}_{\rm full}\xspace}}
\newcommand{\cpart}{\ensuremath{{C}_{\rm part}\xspace}}

\newcommand{\cstar}{\ensuremath{{C}^{*}\xspace}}

\DeclareMathOperator*{\E}{\mathbb{E}}

\newcommand{\set}[1]{{\left\{#1\right\}}}

\newcommand{\full}{{\mathrm{full}}}
\newcommand{\parti}{{\mathrm{part}}}
\newcommand{\opt}{{\mathrm{OPT}}}

\newcommand{\Ball}{{\mathrm{Ball}}}

\graphicspath{{figs/}}

\newcommand{\initOneLiners}{%
    \setlength{\itemsep}{0pt}
    \setlength{\parsep }{0pt}
    \setlength{\topsep }{0pt}
}

\ifdefined\fullversion
\else

\fi

\makeatletter
\newcommand\Label[1]{&\refstepcounter{equation}(\theequation)\ltx@label{#1}&}
\makeatother 

%% file: abstract.tex
\begin{abstract}
	In this paper, we present a new iterative rounding framework for many clustering problems. Using this, we obtain an $(\alpha_1 + \epsilon \leq 7.081 + \epsilon)$-approximation algorithm for $k$-median with outliers, greatly improving upon the large implicit constant approximation ratio of Chen \cite{Chen08}.  For $k$-means with outliers, we give an $(\alpha_2+\epsilon \leq 53.002 + \epsilon)$-approximation, which is the first $O(1)$-approximation for this problem. The iterative algorithm framework is very versatile; we show how it can be used to give $\alpha_1$- and $(\alpha_1 + \epsilon)$-approximation algorithms for matroid and knapsack median problems respectively, improving upon the previous best approximations ratios of $8$~\cite{Swamy16} and $17.46$~\cite{Byrka15}.
	
	The natural LP relaxation for the $k$-median/$k$-means with outliers problem has an unbounded integrality gap. In spite of this negative result, our iterative rounding framework shows that we can round an LP solution to an \emph{almost-integral} solution of small cost, in which we have at most two fractionally open facilities. Thus, the LP integrality gap arises due to the gap between almost-integral and fully-integral solutions. Then, using a pre-processing procedure, we show how to convert an almost-integral solution to a fully-integral solution losing only a constant-factor in the approximation ratio. By further using a sparsification technique, the additive factor loss incurred by the conversion can be reduced to any $\epsilon > 0$.
\end{abstract}

%% file: intro.tex
\section{Introduction}
Clustering is a fundamental task studied by several scientific communities (such as operations research, biology, and of course computer science and machine learning) due to the diversity of its applications. Of the many ways in which the task of clustering can be formalized, the $k$-median and $k$-means problems are arguably the most commonly studied methods by these different communities, perhaps owing to the simplicity of their problem descriptions. In the $k$-median (resp. $k$-means) problem, we are given a set $C$ of clients, a set $F$ of potential facility locations, and a metric space $d(,\cdot,): C \cup F \rightarrow \R_{\geq 0}$, and the goal is to choose a subset $S \subseteq F$ of cardinality at most $k$ so as to minimize $\sum_{j \in C} d(j,S)$ (resp. $\sum_{j \in C} d^2(j,S)$), where $d(j, S):=\min_{i \in S}d(j, i)$.
Both problems are NP-hard in the worst case, and this has led to a long line of research on obtaining efficient approximation algorithms. For the $k$-median problem, the current best known factors are an upper bound of $2.671$ \cite{BPRST17}, and a lower bound of $1+2/e\approx 1.736$ \cite{JMS02}. The $k$-means problem, while originally not as well studied in the theoretical community, is now increasingly gaining attention due to its importance in machine learning and data analysis. The best known approximations are $9$~\cite{ANSW16} and $6.357$~\cite{ANSW16} for general and Euclidean metrics respectively, while the lower bounds are $1+8/e \approx 3.94$ and $1.0013$~\cite{LSW15} for general and Euclidean metrics respectively. Both problems admit PTASs~\cite{ARR98,FRS16,CKM16} on fixed-dimensional Euclidean metrics.

Despite their simplicity and elegance, a significant shortcoming these formulations face in real-world data sets is that they are \emph{not robust to noisy points}, i.e., a few outliers can completely change the cost as well as structure of solutions. To overcome this shortcoming, Charikar et al.~\cite{CKMN01} introduced the robust $k$-median (\kmwo) problem (also called $k$-median with outliers), which we now define.

\begin{definition}[The Robust $k$-Median and $k$-Means Problems] \label{def:kmwo}
The input to the Robust $k$-Median (\kmwo) problem is a set $C$ of clients, a set $F$ of facility locations, a metric space $(C \cup F, d)$, integers $k$ and $m$. The objective is to choose a subset $S \sse F$ of cardinality at most $k$, and a subset $C^* \subseteq C$ of cardinality at least $m$ such that the total cost $\sum_{j \in C^*} d(j,S)$ is minimized. In the Robust $k$-Means (\kmeanswo) problem, we have the same input and the goal is to minimize $\sum_{j \in C^*} d^2(j,S)$.
\end{definition}

The problem is not just interesting from the clustering point of view. In fact, such a joint view of clustering and removing outliers has been observed to be more effective~\cite{CG13,OPRC14} for even the sole task of outlier detection, a very important problem in the real world. Due to these use cases, there has been much recent work~\cite{CG13,GKLMV17,RSKW17} in the applied community on these problems. However, their inherent complexity from the theoretical side is much less understood.
For \kmwo, Charikar et al.~\cite{CKMN01} give an algorithm that violates the number of outliers by a factor of $(1 + \epsilon)$, and has cost at most $4(1 + 1/\epsilon)$ times the optimal cost. Chen \cite{Chen08} subsequently showed a pure $O(1)$-approximation without violating $k$ or $m$. However, Chen's algorithm is fairly complicated
, and the (unspecified) approximation factor is rather large.
 For the \kmeanswo problem, very recently, Gupta et al.~\cite{GKLMV17} gave a bi-criteria $O(1)$-approximation violating the number of outliers, and Friggstad et al.~\cite{FKRS17} give \emph{bi-criteria} algorithms that open $k(1+\epsilon)$ facilities, and have an approximation factor of $(1+\epsilon)$ on Euclidean and doubling metrics, and $25+\epsilon$ on general metrics.

\subsection{Our Results}
In this paper, we present a simple and generic framework for solving clustering problems with outliers, which substantially improves over the existing approximation factors for \kmwo and \kmeanswo.
\begin{theorem}
\label{theorem:kmwo}
	For $\alpha_1 = \inf_{\tau > 1} (3 \tau - 1)/ \ln \tau \leq 7.081$, 
	we have an $\alpha_1(1+\eps)$-approximation algorithm for \kmwo, in running time $n^{O(1/\eps^2)}$.
\end{theorem}

\begin{theorem}
\label{theorem:kmeanswo}
	For $\alpha_2 = \inf_{\tau > 1} \frac{(\tau + 1) (3 \tau - 1)^2}{2 (\tau-1) \ln \tau} \leq 53.002$, 
	we have an $\alpha_2(1 + \eps)$-approximation algorithm for \kmeanswo, in running time $n^{O(1/\eps^3)}$.
\end{theorem}

In fact, our framework can also be used to improve upon the best known approximation factors for two other basic problems in clustering, namely the matroid median (\matmed) and the knapsack median (\knapmed) problems. As in $k$-median, in both problems we are given $F$, $C$ and a metric $d$ over $F \cup C$ and the goal is to select a set $S \subseteq F$ so as to minimize $\sum_{j \in C}d(j, S)$. In \matmed, we require $S$ to be an independent set of a given matroid. In \knapmed, we are given a vector $w \in \R_{\geq 0}^F$ and a bound $B \geq 0$ and we require $\sum_{i \in S} w_i \leq B$. The previous best algorithms had factors $8$ \cite{Swamy16} and $17.46$ \cite{Byrka15} respectively. 

\begin{theorem}
	For $\alpha_1 = \inf_{\tau > 1} (3 \tau - 1)/ \ln \tau \leq 7.081$, we have an efficient $\alpha_1$-approximation algorithm for \matmed, and an $\alpha_1(1+\eps)$-approximation for \knapmed with running time $n^{O(1/\eps^2)}$.
\end{theorem}

\subsection{Our techniques} \label{sec:techniques}


For clarity in presentation, we largely focus on \kmwo in this section.
Like many provable algorithms for clustering problems, our starting point is the natural LP relaxation for this problem. Howeover, unlike the vanilla $k$-median problem, this LP has an unbounded integrality gap (see~\Cref{sec:bad-egs}) for \kmwo. Nevertheless, we can show that all is not lost due to this gap. Indeed, one of our main technical contributions is in developing an iterative algorithm which can round the natural LP to compute an \emph{almost-integral} solution, i.e., one with at most two fractionally open facilities. By increasing the two possible fractional $y$ values to $1$, we can obtain a solution with at most $k+1$ open facilities and cost bounded by $\alpha_1$ times the optimal LP cost. So, the natural LP is good if we are satisfied with such a \emph{pseudo-solution}\footnote{Indeed, an $O(1)$ pseudo-approximation with $k+1$ open facilities is also implicitly given in \cite{Chen08}, albeit with much a larger constant as a bound on the approximation ratio.}. So the unbounded integrality gap essentially arises as the gap between almost-integral solutions and fully-integral ones. In what follows, we first highlight how we overcome this gap, and then give an overview of the rounding algorithm.

\subsubsection{Overcoming Gaps Between Almost-Integral and Integral Solutions}
In the following, let $y_i$ denote the extent to which facility $i$ is open in a fractional solution. Note that once the $y_i$ variables are fixed, the fractional client assignments can be done in a greedy manner until a total of $m$ clients are connected fractionally. Now, assume that we have an almost-integral solution $y$ with two strictly fractional values $y_{i_1}$ and $y_{i_2}$. To round it to a fully-integral one, we need to increase one of $\set{y_{i_1}, y_{i_2}}$ to $1$, and decrease the other to $0$, in a direction that maintains the number of connected clients. Each of the two operations will incur a cost that may be unbounded compared to the LP cost, and they lead to two types of gap instances described in Section~\ref{sec:preprocessing}. We handle these two types separately.

The first type of instances, correspondent to the ``increasing'' operation, arise because the cost of the clients connected to the facility increased to $1$ can be very large compared to the LP cost. We handle this by adding the so-called ``star constraints'' to the natural LP relaxation, which explicitly enforces a good upper bound on the total connection to each facility. The second type of instances, correspondent to the ``decreasing'' operation, arise because there could be many clients very near the facility being set to $0$, all of which now incur a large connection cost. We handle this by a preprocessing step, where we derive an upper bound $R_j$ on the connection distance of each client $j$, which ensures that in any small neighborhood, the number of clients which are allowed to have large connection cost is small. The main challenge is in obtaining good bounds $R_j$ such that there exists a near-optimal solution respecting these bounds.

The above techniques are sufficient to give an $O(1)$ approximation for \kmwo and \kmeanswo. Using an additional sparsification technique,  we can reduce the gap between an almost-integral solution and an integral one to an additive factor of $\eps$. Thus, our final approximation ratio is essentially the factor for rounding an LP solution to an almost-integral one.  The sparsification technique was used by Li and Svensson \cite{LS13} and by Byrka et al.\ \cite{Byrka15} for the $k$-median and knapsack median problems respectively. The detail in how we apply the technique is different from those in \cite{LS13,Byrka15}, but the essential ideas are the same. We guess $O_\eps(1)$ balls of clients that incur a large cost in the optimum solution (``dense balls''), remove these clients,  pre-open a set of facilities and focus on the residual instance. We show that the gap between almost-integral and fully-integral solutions on the residual instance (where there are no dense balls) is at most $\epsilon$.

\subsubsection{Iterative Rounding to Obtain Almost-Integral Solutions} \label{sec:iter-tech}
We now highlight the main ideas behind our iterative rounding framework. 
At each step, we maintain a partition of $C$ into $\cfull$, the set of clients which need to be fully connected, and $\cpart$, the set of partially connected clients. Initially, $\cfull = \emptyset$ and $\cpart = C$. For each client, we also have a set $F_j$ denoting the set of facilities $j$ may connect to, and a ``radius'' $D_j$ which is the maximum connection cost for $j$; these can be obtained from the initial LP solution. Since the client assignments are easy to do once the $y$ variables are fixed, we just focus on rounding the $y$ variables. In addition to $\cfull$, we also maintain a subset \cstar~of full clients such that a) their $F_j$ balls are disjoint, and b) every other full client $j \notin \cstar$~is ``close'' (within $O(1) D_j$) to the set \cstar. Finally, for each client $j$, we maintain an \emph{inner-ball} $B_j$ which only includes facilities in $F_j$ at distance at most $D_j/\tau$ from $j$ for some constant $\tau  > 1$.

We then define the following \emph{core constraints}: (i) $y(F_j) = 1$ for every $j \in C^*$, where $y(S):=\sum_{i \in S}y_i$, (ii) $y(F) \leq k$, and (iii) the total number of connected clients is at least $m$.
As for non-core constraints, we define $y(F_j) \leq 1$ for all partial clients, and $y(B_j) \leq 1$ for all full clients.
These constraints define our polytope ${\mathcal{P}} \subseteq [0,1]^F$.

Then, in each iteration, we update $y$ to be a vertex point in $\mathcal{P}$ that minimizes the linear function $$\textstyle \sum_{j \in \cpart} d(i,j) y_i \allowbreak + \sum_{j \in \cfull} \left( \sum_{i \in B_j} d(i,j) y_i + (1 - y(B_j) ) D_j/\tau \right).$$  Now, if none of the non-core constraints are tight, then this vertex point $y$ is defined by a laminar family of equalities along with a total coverage constraint, which is almost integral and so we output this. Otherwise, some non-core constraint is tight and we make the following updates and proceed to the next iteration. Indeed, if $y(F_j) = 1$ for some $j \in \cpart$, we make it a full client and update $\cstar$ accordingly. If $y(B_j) = 1$ for some $j \in \cfull$, we update its $D_j$ to be $D_j/\tau$ and its $F_j$ to be $B_j$. In each iteration of this rounding, the cost of the LP solution is non-increasing (since $D_j$ and $F_j$ are non-increasing), and at the end, we relate the total connection cost of the final solution in terms of the LP objective using the property that every full client is within $O(D_j)$ from some client in $C^*$.

The iterative rounding framework is versatile as we can simply customize the core constraints.  For example, to handle the the matroid median and knapsack median problems, we can remove the coverage constraint and add appropriate constraints.  In matroid median, we require $y$ to be in the given matroid polytope, while in knapsack median, we add the knapsack constraint. For matroid median, the polytope defined by core constraints is already integral and this leads to our $\alpha_1$-approximation.  For knapsack median, the vertex solution will be almost-integral, and again by using the sparsification ideas we can get our $(\alpha_1+\epsilon)$-approximation. The whole algorithm can be easily adapted to \kmeanswo  to get an $(\alpha_2 + \epsilon)$-approximation.

\subsection{Related Work}
The $k$-median and the related uncapacitated facility location (UFL) problem are two of the most classic problems studied in approximation algorithms. There is a long line of research for the two problems \cite{LV92,STA97,JV99,CS03,KPR98,CG99,JMMSV03,JMS02,MYZ06,Byrka07,CGST99,AGKMMP01} and almost all major techniques for approximation algorithms have been applied to the two problems (see the book of Williamson and Shmoys \cite{WS11}). 
The input of UFL is similar to that of $k$-median, except that we do not have an upper bound $k$ on the number of open facilities, instead each potential facility $i \in F$ is given an opening cost $f_i \geq 0$. The goal  is to minimize the sum of connection costs and facility opening costs.  For the problem, the current best approximation ratio is 1.488 due to Li \cite{Li11} and there is a hardness of 1.463 \cite{GK98}.  For the outlier version of uncapacitated facility, there is a 3-approximation due to \cite{CKMN01}.  This suggests that the outlier version of UFL is easier than that of $k$-median, mainly due to the fact that constraints imposed by facility costs are soft ones, while the requirement of opening $k$ facilities is a hard one.


The $k$-means problem in Euclidean space is the clustering problem that is used the most in practice, and the most popular algorithm for the problem is the Lloyd's algorithm \cite{Lloyd06} (also called ``the $k$-means'' algorithm). However, in general, this algorithm has no worst case guarantee and can also have super-polynomial running time. There has also been a large body of work on bridging this gap, for example, by considering variants of the Lloyd's algorithm \cite{AV07}, or bounding the smoothed complexity~\cite{AMR09}, or by only focusing on instances that are ``stable/clusterable'' \cite{KK10,ABS10,ORSS12,BBG13,CS17}. 

The \matmed problem was first studied by Krishnaswamy et. al. \cite{KKNSS11} as a generalization of the \textit{red-blue median} problem\cite{HKK12}, who gave an $O(1)$ approximation to the problem. Subsequently, there has been work~\cite{CL12,Swamy16} on improving the approximation factor and currently the best upper bound is an $8$-approximation by Swamy \cite{Swamy16}.  As for \knapmed, the first constant factor approximation was given by Kumar \cite{Kumar12}, who showed a $2700$ factor approximation. The approximation factor was subsequently improved by \cite{CL12,Swamy16}, and the current best algorithm is a $17.46$-approximation~\cite{Byrka15}.

\subsection{Paper Outline}
We define some useful notations in Section~\ref{sec:prelim}. Then in Section~\ref{sec:preprocessing}, we explain our preprocessing procedure for overcoming the LP integrality gaps for \kmwo/\kmeanswo. Then in Section~\ref{sec:framework}, we give our main iterative rounding framework which obtains good almost-integral solutions.  Section~\ref{sec:open-k} will show how to covert an almost-integral solution to an integral one, losing only an additive factor of $\eps$. Finally, in~\Cref{sec:matroid,sec:knapsack}, we present our $\alpha_1$ and $(\alpha_1+\eps)$ approximation algorithms for \matmed and \knapmed respectively.

{\em Getting a pseudo-approximation:} if one is only interested in getting a pseudo-approximation for \kmwo/\kmeanswo, i.e, an $O(1)$-approximation with $k + 1$ open facilities, then our algorithm can be greatly simplified.  In particular, the preprocessing step and the conversion step in Sections~\ref{sec:preprocessing} and~\ref{sec:open-k} are not needed, and proofs in Section~\ref{sec:framework} can be greatly simplified.  Such readers can directly jump to Section~\ref{sec:framework} after reading the description of the natural LP relaxation in Section~\ref{sec:preprocessing}. In Section~\ref{sec:framework} we use the term pseudo-approximation setting to denote the setting in which we only need a pseudo-approximation. 

%% file: prelim.tex
\section{Preliminaries}
\label{sec:prelim}

To obtain a unified framework for both \kmwo and \kmeanswo, we often consider an instance $\mathcal{I} = (F, C, d, k, m)$ which could be an instance of either problem, and let a parameter $q$ denote the particular problem: $q=1$ for \kmwo instances and $q=2$ for \kmeanswo instances. Because of the preprocessing steps, our algorithms deal with \emph{extended} \kmwo/\kmeanswo instances, denoted as $(F, C, d, k, \allowbreak m, S_0)$. Here, $F, C, d, k$ and $m$ are defined as before, and $S_0 \subseteq F$ is a set of \emph{pre-opened facilities} that feasible solutions must contain.
We assume that $d(i, i') > 0$ for $i \neq i' \in F$ since otherwise we can simply remove $i'$ from $F$; however, we allow many clients to be collocated, and they can be collocated with one facility.

We use a pair $(S^*, C^*)$ to denote a solution to a (extended) \kmwo or \kmeanswo instance, where $S^* \subseteq F$ is the set of open facilities and $C^* \subseteq C$ is the set of connected clients in the solution, and each $j \in C^*$ is connected to its nearest facility in $S^*$. Note that given $S^*$, the optimum $C^*$ can be found easily in a greedy manner, and thus sometimes we only use $S^*$ to denote a solution to a (extended) \kmwo/\kmeanswo instance.
Given a set $S^* \subseteq F$ of facilities, it will be useful in the analysis to track the nearest facility for every $p \in F \cup C$ and the corresponding distance. To this end, we define the \emph{nearest-facility-vector-pair} for $S^*$ to be $(\kappa^* \in (S^*)^{F\cup C}, c^* \in \R^{F \cup C})$, where $\kappa^*_p = \arg\min_{i \in S^*} d(i, p)$ and $c^*_p = d(\kappa^*_p, p) = \min_{i \in S^*}d(i, p)$ for every $p \in F \cup C$. Though we only connect clients to $S^*$, the distance from a facility to $S^*$ will be useful in our analysis.

We use $y$ (and its variants) to denote a vector in $[0, 1]^F$. For any $S \subseteq F$, we define $y(S) = \sum_{i \in S}y_i$ (same for the variants of $y$). Finally, given a subset  $P \subseteq F \cup C$, a point $p \in F\cup C$ and radius $r \in \R$, we let $\Ball_P(p, r):= \set{p' \in P: d(p, p') \leq r}$ to denote the set of points in $P$ with distance at most $r$ from $p$.

%% file: integral.tex
\section{Preprocessing the \kmwo/\kmeanswo Instance}
\label{sec:preprocessing}
In this section, we motivate and describe our pre-processing step for the \kmwo/\kmeanswo problem, that is used to reduce the integrality gap of the natural LP relaxation. First we recap the LP relaxation and explain two different gap examples which, in turn illustrate two different reasons why this integrality gap arises. Subsequently, we will describe our ideas to overcome these two sources of badness.

\subsection{Natural LP and Gap Examples} \label{sec:bad-egs}
The natural LP for the \kmwo/\kmeanswo problem is as follows.
\begin{equation}
	\min \qquad \sum_{i \in F, j \in C}x_{i, j} d^q(i, j) \qquad \text{s.t.} \tag{$\text{LP}_{\text{basic}}$} \label{LP:basic}
\end{equation} \vspace*{-10pt}

	\begin{minipage}{0.5\linewidth}
		\begin{alignat*}{2}
			\sum_i y_i &\leq k\\
			x_{i, j} &\leq y_i &\quad &\forall i \in F, j \in C
		\end{alignat*}
	\end{minipage}
	\begin{minipage}{0.5\linewidth}
		\begin{alignat*}{2}
			\sum_{i \in F}x_{i, j} &\leq 1 &\quad &\forall j \in C\\
			\sum_{j \in C} \sum_{i \in F}x_{i, j} &\geq m
		\end{alignat*}
	\end{minipage}\smallskip

In the correspondent integer programming, $y_i \in \{0, 1\}$ indicates whether a facility $i \in F$ is open or not, and $x_{i, j} \in \{0, 1\}$ indicates whether a client $j$ is connected to a facility $i$ or not. The objective function to minimize is the total connection cost and all the above constraints are valid. In the LP relaxation, we only require all the variables to be non-negative.

Notice that once the $y$ values are fixed, the optimal choice of the $x$ variables can be determined by a simple greedy allocation, and so we simply state the values of the $y$ variables in the following gap examples. For simplicity, we focus on \kmwo problem when talking about the gap instances.

\begin{figure}[h]
	\centering
\ifdefined\fullversion
	\includegraphics[width=\textwidth]{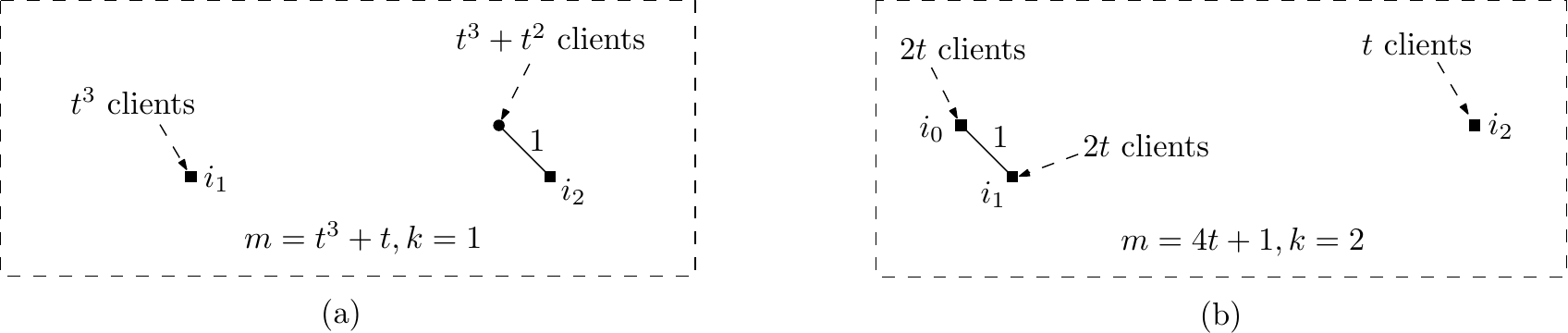}
\else
	\includegraphics[width=0.6\linewidth]{badinstances-V}
\fi
	\caption{Two gap instances.} \label{fig:badinstances} For instance (a), the optimum solution has cost $t^3 + t$. An LP solution with $y_1 = 1-1/t$ and $y_2 = 1/t$ has cost $t^2 + t$.  For instance (b), the optimum solution has cost $t + 1$. An LP solution with $y_0 = 1, y_1 = 1-1/t$ and $y_2 = 1/t$ has cost $2$.
\end{figure}

\medskip \noindent {\bf Gap Example 1.}
Instance (a) in Figure~\ref{fig:badinstances} contains two separate clusters in the metric: in the first cluster, there is a facility $i_1$ collocated with $t^3$ clients, and in the second cluster (very far from the first cluster), there is a facility $i_2$ at unit distance from $t^3 + t^2$ clients. Suppose $k = 1$ and $m = t^3 + t$. Clearly, the integral optimal solution is forced to open a center at $i_2$, thereby incurring a cost of $m = t^3 + t$. However, note that an LP solution can open $1 -1/t$ fraction of facility $i_1$ and $1/t$ fraction of facility $i_2$. This can give a solution with cost $1/t \times (t^3 + t^2) = t^2 + t$, and connecting $(1-1/t)\times t^3 + 1/t \times (t^3 + t^2) = t^3 + t = m$ clients.  Hence the integrality gap is $\frac{t^3+t}{t^2 + t}$, which is unbounded as $t$ increases.

\medskip \noindent {\bf Gap Example 2.}
Instance (b) contains 3 facilities $i_0, i_1$ and $i_2$, collocated with $2t$, $2t$ and $t$ clients respectively. $d(i_0, i_1) = 1$ and $i_2$ is far away from $i_0$ and $i_1$.  We need to open $k = 2$ facilities and connect $m = 4t + 1$ clients. An integral solution opens $i_0$ and $i_2$, connect the $3t$ clients at $i_0$ and $i_2$, and connect $t+1$ clients at $i_1$ to the open facility at $i_0$, incurring a total cost of $t+1$. On the other hand, in an LP solution, we can set $y_0 = 1, y_1 = 1 - 1/t$ and $y_2 = 1/t$. We fully connect the $4t$ clients at $i_0$  and $i_1$, and connect each client at $i_2$ to an extent of $1/t$. 
The total number of connected clients is $4t + t \times 1/t = 4t + 1 = m$. Each client at $i_1$ will be connected to $i_0$ with fraction $1-y_1 = 1/t$, incurring a total cost of $2t \times 1/t = 2$. Thus, the integrality gap is $\frac{t+1}{2}$, which is again unbounded as $t$ increases.



\subsection{Preprocessing}
 In both the examples in~\Cref{sec:bad-egs}, the unbounded integrality gap essentially arises from needing to round an \emph{almost-integral} solution $y$ into a \emph{truly} integral one. This is no co-incidence as we show in Section~\ref{sec:framework} that we can indeed obtain an almost-integral solution with cost at most $\alpha_q$ times the optimal LP cost. Let us now examine the difficulty in rounding these last two fractional variables. As in both instances in Figure~\ref{fig:badinstances}, suppose $i_1$ and $i_2$ are the two facilities with fractional $y$ values and let their $y$ values be $y_1 = 1-\rho$ and $y_2 = \rho$ respectively, where $\rho > 0$ is a sub-constant. A natural idea is to increase one of these variables to $1$ and decrease the other to $0$. Suppose that, in order to maintain the number $m$ of connected clients, we are forced to increase $y_2$ to 1 and decrease $y_1$ to 0 (as in the two gap instances). Then two costs, corresponding to the two gap instances detailed above, will be incurred by this operation. The first cost is incurred by increasing $y_2$ to 1. Some clients were partially connected to $i_2$ in the almost-integral solution. In order to satisfy the coverage requirement, most of these clients need to be fully connected to $i_2$ in the integral solution. The instance in Figure~\ref{fig:badinstances}(a) gives the example where this incurred cost is large compared to the optimal LP value.
The second cost is incurred by decreasing $y_1$ to 0. Some clients were fully connected in the almost-integral solution, each being connected to $i_1$ with fraction $y_1 = 1-\rho$ and to some other facility, say, $i_0 \notin \set{i_1, i_2}$ to extent $1-y_1 = \rho$.  Then, to satisfy the coverage requirement of these clients, we need to connect them to $i_0$ in the integral solution, and this cost could be large (see Figure~\ref{fig:badinstances}(b)).

\medskip \noindent {\bf Preprocessing Idea 1: Bounding Star Cost.} \label{sec:star}
The first kind of bad instances described above are fairly straightforward to handle. Indeed, when increasing $y_2$ from $\rho$ to $1$, the incurred cost is at most the cost of the ``star'' associated with $i_2$, i.e., the total connection cost of clients which are fractionally connected to $i_2$. We know the cost of such a star in the optimum solution is at most $\opt$ and thus we can add the following constraints to overcome this type of bad instances:  $\sum_{j \in C}x_{i, j}d^q(i, j) \leq y_i \cdot \opt$ for every $i \in F$. 

In fact, we can bring down this additive error to $\rho \cdot \opt$ (where $\rho \to 0$ as $\eps \to 0$) by \emph{guessing} the set $S_1$ of centers corresponding to the expensive stars (whose connection cost exceeds $\rho \cdot \opt$) in the optimum solution and opening them. 
%
More importantly, we can strengthen \eqref{LP:basic} by enforcing the constraints bounding the connection cost to facilities $i \notin S_1$:
$\sum_{j\in C}x_{i,j}d^q(i,j) \leq y_i \cdot \rho \cdot \opt $.

\medskip \noindent {\bf Preprocessing Idea 2: Bounding Backup Cost.} \label{sec:backup}
Let us now explain how we handle the large incurred cost when we decrease $y_1$ to 0, as in Figure~\ref{fig:badinstances}(b).  
Note that in the optimal solution, for any $R \geq 0$, the number of clients whose connection cost is at least $R$ is at most $\opt/R$. In particular, in the example in Figure~\ref{fig:badinstances}(b),
if $\opt$ is indeed $\Theta(1)$, then the number of clients at $i_1$ with connection cost at least $1$ can be at most $\opt$. To this end, suppose we are able to \emph{guess} a \emph{specified} set of $\opt = \Theta(1)$ clients located at $i_1$, and \emph{disallow} the remaining collocated clients from connecting to facilities which are at distance $1$ or greater. Then we could set $x_{ij} = 0$ for all disallowed connections which in turn will make the LP infeasible when our guess for $\opt$ is $\Theta(1)$ for the bad instance from  Figure~\ref{fig:badinstances}(b).
 While it is difficult to get such good estimates on the disallowed connections of clients on a global scale, we show that we can indeed make such restrictions \emph{locally}, which is our second pre-processing idea. 
We show that we can efficiently obtain a vector $(R_j)_{j \in C}$ of upper bounds on connection costs of clients which satisfies the following two properties: (i) for any $R$ and some constant $\delta > 0$, the number of clients in \emph{any ball} of radius $\delta R$ with $R_j$ values more than some $\Theta(R)$ is at most $O(1) \opt/R^q$, and (ii) there exists a feasible solution of cost at most $O(1) \opt$ that respects these upper bounds $(R_j)_{j \in C}$ on client connections. This will then help us bound the re-assignment cost incurred by decreasing $y_1$ to $0$ (i.e., shutting down facility $i_1$).


\medskip \noindent {\bf Preprocessing Idea 3: Sparsification.} \label{sec:sparse} To reduce the additive factor we lose for handling the second type of bad instances to any small constant $\epsilon$, we use the so-called \emph{sparsification} technique.  Now, 
suppose we had the situation that,  in the optimum solution,  the number of clients in any ball of radius $\delta R$ with connection cost at least $R$ is miraculously at most $\rho \cdot \opt/R^q$ for some small $0 \leq \rho < 1$. Informally, we call such instances as sparse instances. Then, we can bound total increase in cost by re-assigning these clients when shutting down one facility by $O(\rho) \opt$. Indeed, our third pre-processing idea is precisely to correctly guess such ``dense'' regions (the clients and the facility they connect to in an optimum solution) so that the remaining clients are locally sparse. We note that this is similar to ideas employed in previous algorithms for $k$-median and \knapmed.

\medskip

Motivated by these ideas, we now define the notion of \emph{sparse extended instances} for \kmwo/\kmeanswo.

\begin{definition}[Sparse Instances]
	\label{def:sparse}
	Suppose we are given an extended \kmwo/\kmeanswo instance ${\mathcal{I}'} = (F, C', d, k, m', S_0)$, and parameters $\rho, \delta \in (0, 1/2)$ and $U \geq 0$. Let $(S^*, C'^*)$ be a solution to $\mathcal{I}'$ with cost at most $U$, and $(\kappa^*, c^*)$ be the nearest-facility-vector-pair for $S^*$. Then we say ${\mathcal{I}'}$ is $(\rho,\delta, U)$-sparse w.r.t the solution $(S^*, C'^*)$, if
	\begin{enumerate}[itemsep=0pt,topsep=3pt,label=(\ref{def:sparse}\alph*),leftmargin=*]
		\item \label{property:sparse-star-cheap} for every $i \in S^* \setminus S_0$, we have $\sum_{j \in C'^*:\kappa^*_j = i}(c^*_j)^q \leq \rho U$,
		\item \label{property:sparse-sparse} for every $p \in F \cup C'$, we have $\big|\Ball_{C'^*}(p, \delta c^*_p)\big|\cdot (c^*_p) ^q\leq \rho U$.
	\end{enumerate}
\end{definition}
Property~\ref{property:sparse-star-cheap} requires the cost incurred by each $i \in S^* \setminus S_0$ in the solution $(S^*, C'^*)$ is at most $\rho U$. Property~\ref{property:sparse-sparse} defines the sparsity requirement: roughly speaking, for every $p  \in F \cup C'$, in the solution $(S^*, C'^*)$, the total connection cost of clients in $C'^*$ near $p$ should be at most $\rho U$.

We next show in the following theorem, that we can effectively reduce a general \kmwo/\kmeanswo instance to a sparse extended instance with only a small loss in the approximation ratio.

\begin{theorem}
	\label{thm:reduce-to-sparse-instances}
	Suppose we are given a \kmwo/\kmeanswo instance ${\mathcal{I}} = (F, C, d, k, m)$, also parameters $\rho, \delta \in (0, 1/2)$ and an upper bound $U$ on the cost of the optimal solution $(S^*, C^*)$ to $\mathcal{I}$ (which is not given to us). 
	Then there is an $n^{O\left(1/\rho\right)}$-time algorithm that outputs $n^{O\left(1/\rho\right)}$ many extended \kmwo/\kmeanswo instances such that one of the instances in the set, say, ${\mathcal {I'}}$,  has the form $(F, C' \subseteq C, d, k, m' = |C^* \cap C'|, S_0 \subseteq S^*)$ and satisfies:
		\begin{enumerate}[itemsep=0pt,topsep=3pt,label=(\ref{thm:reduce-to-sparse-instances}\alph*),leftmargin=*]
			\item \label{property:reduce-instance-sparse} $\mathcal{I'}$ is $(\rho, \delta, U)$-sparse w.r.t the solution $(S^*, C^* \cap C')$,
			\item \label{property:reduce-instance-good} $\frac{(1-\delta)^q}{(1+\delta)^q}\sum_{j \in C^*\setminus C'}d^q(j, S_0) + \sum_{j \in C^* \cap C'} d^q(j, S^*) \leq U$.
		\end{enumerate}
\end{theorem}

Before we give the proof, we remark that Property~\ref{property:reduce-instance-good} means that the cost of the solution $(S^*, C^* \cap C')$ for the residual sparse instance $\mathcal {I'}$ plus the approximate cost of reconnecting the $m - m'$ clients $C^* \setminus C'$ to the guessed facilities $S_0$ is upper bounded by $U$. 
\begin{proof}
	Let us first assume we know the optimum solution $(S^*, C^*)$ to the instance $\mathcal{I}=(F, C, d, k, m)$, and we will dispense with this assumption later. Let $(\kappa^*, c^*)$ be the nearest-facility-vector-pair for $S^*$. We initialize $C' = C$ and $S_0 = \emptyset$, and construct our instance iteratively until it becomes $(\rho,\delta,U)$-sparse w.r.t $(S^*, C^* \cap C')$. Our instance is always defined as $\mathcal{I'} =(F, C', d, k, m'= |C^* \cap C'|, S_0)$.

Indeed, to satisfy Property~\ref{property:sparse-star-cheap}, we add to $S_0$ the set of facilities $i \in S^*$ with $\sum_{j \in C^*: \kappa^*_j = i} (c^*_j)^q > \rho U$. There are at most $1/\rho$ many such facilities. After this, we have $\sum_{j \in C^* \cap C':\kappa^*_j = i}(c^*_j)^q\leq \sum_{j \in C^*:\kappa^*_j = i}(c^*_j)^q \leq \rho U$ for every $i \in S^* \setminus S_0$.  This will always be satisfied since we shall only add facilities to $S_0$.

To guarantee Property~\ref{property:sparse-sparse}, we run the following iterative procedure: while there exists $p \in F \cup C$ such that $\big|\Ball_{C^*\cap C'}(p, \delta c^*_p)\big|\cdot (c^*_p)^q > \rho U$, we update $S_0 \gets S_0 \cup \set{\kappa^*_p}$ and $C' \gets C' \setminus \Ball_{C'}(p, \delta c^*_p)$. By triangle inequality, each client $j$ removed from $C'$ has $c^*_j \geq c^*_p - d(j, p) \geq (1-\delta)c^*_p$. Also, $d(j, S_0) \leq d(j, \kappa^*_p) \leq d(p, \kappa^*_p) + d(j, p) \leq c^*_p + \delta c^*_p \leq \frac{1+\delta}{1-\delta}c^*_j$. Moreover, the total $(c^*_j)^q$ over all clients $j \in \Ball_{C' \cap C^*}(p, \delta c^*_p)$ is at least $\big|\Ball_{C^*\cap C'}(p, \delta c^*_p)\big|\cdot (1-\delta)^q(c^*_p)^q > (1-\delta)^q\rho U$. Thus, the procedure will terminate in less than $\frac{1}{\rho(1-\delta)^q} = O(1/\rho)$ iterations.  
	After this procedure, we have $|\Ball_{C^* \cap C'}(p, \delta c^*_p)|\cdot (c^*_p)^q \leq \rho U$ for every $p \in F \cup C$. That is, Property~\ref{property:sparse-sparse} holds for the instance $\mathcal{I'} = (F, C', d, k, |C^* \cap C'|, S_0)$ w.r.t solution $(S^*, C^* \cap C')$. Thus Property~\ref{property:reduce-instance-sparse} holds.
	
	Now we prove Property~\ref{property:reduce-instance-good}.
\ifdefined\fullversion
	\begin{align*}
		\frac{(1-\delta)^q}{(1+\delta)^q}\sum_{j \in C^* \setminus C'} d^q(j, S_0) + \sum_{j \in C^* \cap C'} d^q(j, S^*)
		\leq \sum_{j \in C^* \setminus C'} (c^*_j)^q + \sum_{j \in C^* \cap C'} (c^*_j)^q
		= \sum_{j \in C^*} (c^*_j)^q \leq U.
	\end{align*}
\else
	\begin{align*}
		&\frac{(1-\delta)^q}{(1+\delta)^q}\sum_{j \in C^* \setminus C'} d^q(j, S_0) + \sum_{j \in C^* \cap C'} d^q(j, S^*) \\
		&\leq \sum_{j \in C^* \setminus C'} (c^*_j)^q + \sum_{j \in C^* \cap C'} (c^*_j)^q 
		= \sum_{j \in C^*} (c^*_j)^q \leq U.
	\end{align*}
\fi
	The first inequality used the fact that for every $j \in C^* \setminus C'$, we have $d(j, S_0) \leq \frac{1+\delta}{1-\delta}c^*_j$.  Thus Property~\ref{property:reduce-instance-good} holds.
	
	Since we do not know the optimum solution $(S^*, C^*)$ for $\mathcal{I}$, we can not directly apply the above procedure. However, note that $C'$ is obtained from $C$ by removing at most $O(1/\rho)$ balls of clients, $S_0$ contains at most $O(1/\rho)$ facilities, and there are $m$ possibilities for $m' = |C^* \cap C'|$. Thus, there are at most $n^{O\left(1/\rho\right)}$ possible instances $(F, C', d, k, m', S_0)$, and we simply output all of them.
\end{proof}

We next show how, for $(\rho,\delta,U)$-sparse instances, we can effectively guess good upper bounds $R_j$ on the maximum connection cost of each client. The notion of sparse instances is only to get this vector of $R_j$'s and will not be needed after the proof of this theorem.

\begin{theorem}
	\label{thm:preprocessing}
	Let ${\mathcal{I'}} = (F, {C'}, d, k, {m'}, S_0)$ be a $(\rho, \delta, U)$-sparse instance w.r.t some solution $(S^* \supseteq S_0, {C'}^*)$ of ${\mathcal{I'}}$, for some $\rho, \delta \in (0, 1/2)$ and $U \geq 0$. Let $U' \leq U$ be the cost of $(S^*, {C'}^*)$ to ${\mathcal{I'}}$. Then given $\calI', \rho, \delta, U$, we can efficiently find a vector $R = (R_j)_{j \in {C'}} \in \R_{\geq 0}^{C'}$, such that:
	\begin{enumerate}[topsep=3pt, itemsep=0pt, label=(\ref{thm:preprocessing}\alph*),leftmargin=*]
		\item \label{property:prep-sparse} for every $t > 0$ and $p \in F \cup {C'}$, we have
		\begin{align}
			\left|\set{j \in \Ball_{C'}\left(p, \frac{\delta t}{4 + 3\delta}  \right): R_j \geq t} \right| \leq \frac{\rho (1 + 3\delta/4)^q}{(1-\delta/4)^q} \cdot \frac{U}{t^q}, \label{inequ:prep-sparse}
		\end{align}
		\item \label{property:prep-good} there exists a solution to ${\mathcal{I'}}$ of cost at most $(1+\delta/2)^qU'$ where if $j$ is connected to $i$ then  $d(i, j)\leq R_j$; moreover, the total cost of clients connected to any facility $i \notin S_0$ in this solution is at most $\rho(1+\delta/2)^q U$.
	\end{enumerate}
\end{theorem}

Property~\ref{property:prep-sparse} says that in a ball of radius $\Theta(t)$, the number of clients $j$ with $R_j \geq t$ is at most $O(\rho) \cdot \frac{U}{t^q}$. The first part of Property~\ref{property:prep-good} says that there is a near-optimal solution which respects these upper bounds $R_j$ on the connection distances of all clients $j$, and the second part ensures that all the star costs for facilities in $F \setminus S_0$ are still bounded in this near-optimal solution (akin to Property~\ref{property:sparse-star-cheap}).

\begin{proof}[Proof of Theorem~\ref{thm:preprocessing}]
We will now show that the following algorithm efficiently finds a vector $(\widehat{R}_j)_{j \in {C'}} \in \R_{\geq 0}^{C'}$, such that the following properties hold. 
	\begin{enumerate}[topsep=3pt, itemsep=0pt, label=(\roman*), leftmargin=*]
		\item \label{property:proof-prep-sparse} for every $t > 0$ and $p \in F \cup {C'}$, we have
		\begin{align}
			\left|\set{j \in \Ball_{C'}\left(p, \frac{\delta t}{4}  \right): \widehat{R}_j \geq t} \right| \leq \frac{\rho}{(1-\delta/4)^q} \cdot \frac{U}{t^q}, \label{inequ:proof-prep-sparse}
		\end{align}
		\item \label{property:proof-prep-good} there exists a solution to ${\mathcal{I'}}$ of cost at most $(1+\delta/2)^qU'$ where if $j$ is connected to $i$ then  $d(i, j)\leq (1+ 3\delta/4) \widehat{R}_j$; moreover, the total cost of clients connected to any facility $i \notin S_0$ in this solution is at most $\rho(1+\delta/2)^q U$.
	\end{enumerate}
The proof then follows by setting $R_j = \widehat{R}_j (1 + 3\delta/4)$.

	\begin{algorithm}[h]
		\caption{Construct $\widehat{R}$} \label{alg:construct-R}
		\begin{algorithmic}[1]
			\State \textbf{for} every $j \in {C'}$ \textbf{do}: $\widehat{R}_j \gets 0$
			\For{every $t' \in \set{d(i, j): i \in F, j \in {C'}}\setminus \{0\}$, in decreasing order}
				\For{every $j' \in {C'}$ such that $\widehat{R}_{j'} = 0$}
					\If{letting $\widehat{R}_{j'} = t'$ will not violate \eqref{inequ:proof-prep-sparse} for $t = t'$ and any $p \in F \cup {C'}$}
						\State $\widehat{R}_{j'} \gets t'$ \label{state:prep-update-R}
					\EndIf
				\EndFor
			\EndFor
		\end{algorithmic}
	\end{algorithm}

	Consider Algorithm~\ref{alg:construct-R} that constructs the vector $\widehat{R}$. Property~\ref{property:proof-prep-sparse} holds at the beginning of the procedure since all clients $j$ have $\widehat{R}_j = 0$. We show that the property is always maintained as the algorithm proceeds. To this end, focus on the operation in which we let $\widehat{R}_{j'} = t'$ in Step~\ref{state:prep-update-R}.
	\begin{itemize}[topsep=3pt,itemsep=0pt]
		\item This will not violate \eqref{inequ:proof-prep-sparse} for $t = t'$ and any $p \in F \cup {C'}$ as guaranteed by the condition.
		\item It will not violate \eqref{inequ:proof-prep-sparse} for any $t > t'$ and $p \in F \cup {C'}$ since on the left side of \eqref{inequ:proof-prep-sparse} we only count clients $j$ with $\widehat{R}_j \geq t$ and setting $\widehat{R}_{j'}$ to $t' < t$ does not affect the counting.
		\item Focus on some $t < t'$ and $p \in F \cup {C'}$. Firstly, note that $\big|\Ball_{C'}(p, \delta t/4):\widehat{R}_j \geq t\big| \leq \big|\Ball_{C'}(p, \delta t'/4):\widehat{R}_j \geq t \big| = \big|\Ball_{C'}(p, \delta t'/4):\widehat{R}_j \geq t'\big|$. The first inequality holds because we consider a bigger ball on the RHS, and the following equality holds since at this moment any $j$ has either $\widehat{R}_j \geq t'$ or $\widehat{R}_j = 0$, which implies $\widehat{R}_j \geq t \Leftrightarrow \widehat{R}_j \geq t'$. 
		Now, since \eqref{inequ:proof-prep-sparse} holds for $t'$ and $p$, we have
\ifdefined\fullversion		
		\begin{align*}
			\left|\Ball_{C'}\left(p,  \delta t/ 4\right):\widehat{R}_j \geq t\right| &\leq \left|\Ball_{C'}\left(p, \delta t'/4 \right):\widehat{R}_j \geq t'\right|
			\leq \frac{\rho}{(1-\delta/4)^q} \cdot \frac{U}{(t')^q} \leq  \frac{\rho}{(1-\delta/4)^q} \cdot \frac{U}{t^q}.
		\end{align*}
\else
		\begin{align*}
			&\quad \left|\Ball_{C'}\left(p,  \delta t/ 4\right):\widehat{R}_j \geq t\right| \leq \left|\Ball_{C'}\left(p, \delta t'/4 \right):\widehat{R}_j \geq t'\right|\\
			&\leq \frac{\rho}{(1-\delta/4)^q} \cdot \frac{U}{(t')^q} \leq  \frac{\rho}{(1-\delta/4)^q} \cdot \frac{U}{t^q}.
		\end{align*}
\fi
		 Thus, \eqref{inequ:proof-prep-sparse} also holds for this $t$ and $p$.		
	\end{itemize}
	
	This finishes the proof of Property~\ref{property:proof-prep-sparse}. Now consider Property~\ref{property:proof-prep-good}. Let $(\kappa^*, c^*)$ be the nearest-facility-vector-pair for $S^*$. Notice that the cost of the $(S^*, {C'}^*)$ is at most $U'$ as specified in the lemma statement. We build our desired solution incrementally using the following one-to-one mapping $f: {C'}^* \rightarrow C'$ such that the final set of clients which will satisfy~Property~\ref{property:proof-prep-good} will be $f({C'}^*)$. Initially, let $f(j) = j$ for every $j \in {C'}^*$. To compute the final mapping $f$, we apply the following procedure. At a high-level, if a client $j \in {C'}^*$ has connection cost $c^*_j > (1+3\delta/4) \widehat{R}_j$, then it means that there is a nearby ball with many clients with $\widehat{R}$ value at least $\widehat{R}_j$. We show that, since the instance is sparse, at least one of them is not yet mapped in $f({C'}^*)$ and has a larger radius bound, and so we set $f(j)$ to $j'$.

Formally, for every client $j \in {C'}^*$, in non-decreasing order of $c^*_j$, if $c^*_j > (1+3 \delta/4) \widehat{R}_j$, we update $f(j)$ to be a client in ${C'} \setminus f({C'}^*)$ such that
	\begin{enumerate}[topsep=3pt,itemsep=0pt,label=(\Alph*)]
		\item $d(f(j), j) \leq \delta c^*_j/2$, and
		\item $\widehat{R}_{f(j)} \geq c^*_j$.
	\end{enumerate}
Assuming we can find such an $f(j)$ in every iteration, then at the end of the procedure, we have that $f$ remains one-to-one. As for connections, for every $j \in {C'}^*$, let us connect $f(j)$ to $\kappa^*_j$. Then, we have $d(f(j), \kappa^*_j) \leq d(j, \kappa^*_j) + d(j, f(j)) \leq (1+\delta/2)c^*_j$. Now it is easy to see that Property~\ref{property:proof-prep-good} is satisfied. Indeed, the cost of the solution $(S^*, f({C'}^*))$ is clearly at most $(1+\delta/2)^qU'$, and the connection satisfies $d(f(j), \kappa^*_j) \leq (1 + 3\delta/4) \widehat{R}_{f_j}$. Moreover, for every facility $i \in S^* \setminus S_0$, the total cost of clients connected to $i$ is at most $(1+\delta/2)^q\rho U$.

Thus, our goal becomes to prove that  in every iteration we can always find an $f_j$ satisfying properties (A) and (B) above. To this end, suppose we have succeeded in all previous iterations and now we are considering $j \in {C'}^*$ with $c^*_j> (1+3 \delta/4)  \widehat{R}_j$. We thus need to update $f(j)$. Notice that at this time, we have $d(j', f({j'}))\leq \delta c^*_{j'}/2$ for every $j' \in {C'}^*$, by the order we consider clients and the success of previous iterations. Now, focus on the iteration $t' = c^*_j$ in Algorithm~\ref{alg:construct-R}. Then the fact that $c^*_j> (1+3 \delta/4)  \widehat{R}_j$ means that $j$ is not assigned a radius before and during iteration $t'$. Thus, there must be some $p \in F \cup {C'}$, such that $d(p, j) \leq \frac{\delta c^*_j}{4}$, and the set $H_j := \set{j' \in \Ball_{C'}\left(p, \delta c^*_j/4 \right):\widehat{R}_{j'} \geq c^*_j}$ such that $H_j \cup \{j\}$ has cardinality strictly more than $\frac{\rho}{(1-\delta/4)^q}\cdot\frac{U}{(c^*_j)^q}$. We claim that this set $H_j$ has a free client which $j$ can map to. Indeed, if there exists $j' \in H_j$ such that $j' \notin f({C'}^*)$, then we can set $f(j) = j'$. This  satisfies property (A) since $d(j',j) \leq d(j,p) + d(j',p) \leq \delta c^*_j /2$, and property (B) trivally, from the way $H_j$ is defined.

 We now show there exists a $j' \in H_j \setminus f({C'}^*)$, by appealing to the sparsity of the instance. Indeed, suppose $H_j \sse f({C'}^*)$ for the sake of contradiction.
 Now, because $d(j', f_{j'}) \leq \delta c^*_j/2$ for every $j' \in {C'}^*$ at this time, we get that every client in $f^{-1}(H_j)$ has distance at most $\delta c^*_j/4 + \delta c^*_j/2 = 3\delta c^*_j/4$ from $p$. Thus,
	\begin{align*}
		\big|\Ball_{{C'}^*}(p, 3\delta c^*_j/4)\big| \geq |f^{-1}(H_j)| = |H_j| \geq \frac{\rho}{(1-\delta/4)^q}\cdot\frac{U}{(c^*_j)^q}.
	\end{align*}
	By triangle inequality, note that $c^*_p \geq c^*_j - d(j, p) \geq (1-\delta/4)c^*_j \geq (3 \delta/4 ) c^*_j$, as $\delta < 1$. Hence, we have
\ifdefined\fullversion
	\begin{align*}
		&\big|\Ball_{{C'}^*}(p, \delta c^*_p)\big|\cdot (c^*_p)^q
		> \frac{\rho}{(1-\delta/4)^q}\cdot\frac{U}{(c^*_j)^q} \cdot (1-\delta/4)^q(c^*_j)^q = \rho U,
	\end{align*}
\else
	\begin{align*}
		\big|\Ball_{{C'}^*}(p, \delta c^*_p)\big|\cdot (c^*_p)^q
		&> \frac{\rho}{(1-\delta/4)^q}\cdot\frac{U}{(c^*_j)^q} \cdot (1-\delta/4)^q(c^*_j)^q\\ &= \rho U,
	\end{align*}
\fi
	contradicting Property~\ref{property:sparse-sparse} of the $(\rho,\delta,U)$-sparsity of ${\mathcal {I}}'$.
\end{proof}

\subsection{Putting Everything Together}
Now we put everything together to prove~\Cref{theorem:kmwo,theorem:kmeanswo}. We will use, as a black-box, the following theorem, which will be the main result proved in~\Cref{sec:framework,sec:open-k}.
\begin{theorem}[Main Theorem]
	\label{thm:main}
	Let $\mathcal{I'} = (F, C', d, k, m', \allowbreak S_0)$ be an extended \kmwo/\kmeanswo instance, $\rho, \delta \in (0, 1/2)$, $0 \leq U' \leq U$ and $R \in \R_{\geq 0}^{C'}$ such that Properties~\ref{property:prep-sparse} and \ref{property:prep-good} hold. Then there is an efficient randomized algorithm that, given $\calI', \rho, \delta, U$ and $R$, computes a solution to $\mathcal{I'}$ with expected cost at most $\alpha_q  (1 + \delta/2)^q U' + O(\rho/\delta^q)U$.
\end{theorem}


Given a \kmwo/\kmeanswo instance $\mathcal{I} = (F, C, d, k, m)$ and $\epsilon > 0$, we assume (by a standard binary-search idea) that we are given an upper bound $U$ on the cost of the optimum solution $(S^*, C^*)$ to $\mathcal{I}$ and our goal is to find a solution of cost at most $\alpha_q (1 + \epsilon)U$.  To this end, let $\delta = \Theta({\epsilon})$ and let $\rho = \Theta(\epsilon^{q+1})$ be such that $(1+\delta/2)^q \leq 1 + \epsilon/2$ and the $O(\rho/\delta^q)$ term before $U$ (in~\Cref{thm:main}) is at most $\epsilon$.

Using Theorem~\ref{thm:reduce-to-sparse-instances}, we obtain $n^{O(1/\rho)}$ many extended instances. We run the following procedure for each of them, and among all the computed solutions, we return the best one that is valid. Thus, for the purpose of analysis, we can assume that we are dealing with the instance $\mathcal{I'} = (F, C' \subseteq C, d, k, m'=|C^* \cap C'|, S_0 \subseteq S^*)$ satisfying properties~\ref{property:reduce-instance-sparse} and~\ref{property:reduce-instance-good} of Theorem~\ref{thm:reduce-to-sparse-instances}. Defining $C'^* = C^* \cap C'$ and applying~\Cref{thm:preprocessing, thm:main} in order, we obtain a solution $(\tilde S \supseteq S_0, \tilde C \subseteq C')$ to $\mathcal{I'}$ whose expected cost is $\alpha_q (1+\delta/2)^q U' + O(\rho/\delta^q) U \leq \alpha_q(1+\epsilon/2)U' + \epsilon U$.  To extend the solution be feasible for $\mathcal{I}$, we simply greedily connect the cheapest $m - m'$ clients in $C \setminus C'$; the connection cost of these clients is at most $\sum_{j \in C^* \setminus C'}d^q(j, S_0)$. The expected cost of this solution is then at most
$\sum_{j \in C^* \setminus C}d^q(j, S_0) + \alpha_q(1+\epsilon/2)U' + \epsilon U$, which in turn is at most
\ifdefined\fullversion
\begin{align*}
	&\leq \max\set{\frac{(1+\delta)^q}{(1-\delta)^q}, \alpha_q(1+\epsilon/2)} \cdot \left(\frac{(1-\delta)^q}{(1+\delta)^q}\sum_{j \in C^* \setminus C'}d^q(j, S_0) + \sum_{j \in C^* \cap C'}d^q(j, S^*)\right) + \epsilon U\\
	& \leq \alpha_q(1+\epsilon/2)U + \epsilon U \leq \alpha_q(1+\epsilon) U.
\end{align*}
\else
\begin{align*}
	&\leq \max\set{\frac{(1+\delta)^q}{(1-\delta)^q}, \alpha_q(1+\epsilon/2)} \\
	&\cdot \left(\frac{(1-\delta)^q}{(1+\delta)^q}\sum_{j \in C^* \setminus C'}d^q(j, S_0) + \sum_{j \in C^* \cap C'}d^q(j, S^*)\right) + \epsilon U\\
	& \leq \alpha_q(1+\epsilon/2)U + \epsilon U \leq \alpha_q(1+\epsilon) U.
\end{align*}
\fi
The second inequality uses Property~\ref{property:reduce-instance-good} and $\frac{(1+\delta)^q}{(1-\delta)^q} \leq \alpha_q$. The overall running time is $n^{O(1/\rho)} = n^{O(1/\epsilon^{q+1})}$, which is $n^{O(1/\epsilon^2)}$ for $q=1$ and $n^{O(1/\epsilon^3)}$ for $q=2$.

%% file: framework.tex
\section{The Iterative Rounding Framework}
\label{sec:framework}

This section and the next one are dedicated to the proof of the main theorem (Theorem~\ref{thm:main}). \emph{For notational convenience, we replace the $\mathcal{I'}, C', {C'}^*$ and $m'$ with $\mathcal{I}, C, C^*$ and $m$ subsequently.
For clarity in presentation, we restate what we are given and what we need to prove after the notational change, without referring to properties stated in Section~\ref{sec:preprocessing}.} The input to our algorithm is an instance $\mathcal{I} = (F, C, d, k, m, S_0)$, parameters $\rho, \delta \in (0, 1/2), U \geq 0$ and a vector $R \in \R_{\geq 0}^C$. There is a parameter $U' \in [0, U]$ (which is not given to our algorithm) such that the following properties are satisfied:
\begin{enumerate}[topsep=3pt, itemsep=0pt, label=(\ref{thm:preprocessing}\alph*'),leftmargin=*]
	\item \label{property:prep-sparse-restated} for every $t > 0$ and $p \in F \cup C$, we have
	\begin{align}
		\left|\set{j \in \Ball_{C}\left(p, \frac{\delta t}{4 + 3\delta}  \right): R_j \geq t} \right| \leq \frac{\rho (1 + 3\delta/4)^q}{(1-\delta/4)^q} \cdot \frac{U}{t^q}, \label{inequ:prep-sparse-restated}
	\end{align}
	\item \label{property:prep-good-restated} there exists a solution to ${\mathcal I}$ of cost at most $(1+\delta/2)^qU'$ where if $j$ is connected to $i$ then  $d(i, j)\leq R_j$; moreover, the total cost of clients connected to any facility $i \notin S_0$ is at most $\rho(1+\delta/2)^q U$.
\end{enumerate}
Our goal is to output a solution to $\calI$ of cost at most $\alpha_q(1+\delta/2)^q U' + O(\rho/\delta^q)U$. We do this in two parts: in this section, we use our iterative rounding framework and obtain an almost-integral solution for $\mathcal{I}$, with cost $\alpha_q(1+\delta/2)^q U'$. This in fact immediately also gives us the desired pseudo-approximation solution.
In the next section, we show how to convert the almost-integral solution to an integral one, with an additional cost of $O(\rho/\delta^q)U$.

We remark that in the pseudo-approximation setting, many parameters can be eliminated: we can set
$S_0 = \emptyset, U = \infty, \delta = 0,  R_j = \infty$ for every $j \in C$, and $\rho$ to be any positive number.  Then Property~\ref{property:prep-sparse-restated} trivially holds as $U = \infty$. Property~\ref{property:prep-good-restated} simply says that there is an integral solution to $\calI$ with cost at most $U'$.  Our goal is to output an almost-integral solution of cost at most $\alpha_q U'$.




\subsection{The Strengthened LP}
For notational simplicity, we let $\tilde U = (1+\delta/2)^qU$ and $\tilde U' = (1+\delta/2)^q {U'}$. We now present our strengthened LP, where we add constraints~\eqref{LPs:S0}-\eqref{LPs:R-d} to the basic LP.
\ifdefined\fullversion
	\begin{align*}
		\min \qquad \sum_{i \in F, j \in C}x_{i, j} d^q(i, j) \tag{$\text{LP}_{\text{strong}}$} \qquad \text{s.t. constraints in \eqref{LP:basic} and} \label{LP:strong}
	\end{align*}\vspace*{-20pt}
\else
	\begin{align*}
		\min \qquad \sum_{i \in F, j \in C}x_{i, j} d^q(i, j) \tag{$\text{LP}_{\text{strong}}$} \label{LP:strong} \\
		\qquad \text{s.t. constraints in \eqref{LP:basic} and}
	\end{align*}
\fi
	\begin{align}
		y_i &= 1 & & \forall i \in S_0 &\label{LPs:S0}\\
		 x_{i,j} &= 0 & &\forall i, j \, \text { s.t } \, d(i,j) >  R_j & \label{LPs:R} \\
	    \sum_{j} d^q(i,j) x_{i,j} & \leq \rho \tilde U y_i & & \forall i \notin S_0 & \label{LPs:star} \\[-6pt]
	    x_{i,j} &= 0 & &\forall i \notin S_0, j \, \text { s.t } \, d^q(i,j) >  \rho \tilde U &\label{LPs:R-d}
	\end{align}

Note that in the pseudo-approximation setting, we do not have Constraints~\eqref{LPs:S0} to ~\eqref{LPs:R-d}, since $S_0 = \emptyset, R_j = \infty$ for all $j \in C$ and $\tilde U = \infty$. Thus \eqref{LP:strong} is the same as \eqref{LP:basic}.
\begin{lemma}
	The cost of the optimum solution to~\eqref{LP:strong} is at most $\tilde U' = (1+\delta/2)^q{U'}$.
\end{lemma}

\begin{proof}
Consider the solution satisfying Property~\ref{property:prep-good-restated}, and set $y_i = 1$ if $i$ is open and $x_{i, j} = 1$ if $j$ is connected to $i$ in the solution. Then the constraints in \eqref{LP:basic} and \eqref{LPs:S0} hold as we have a valid solution. Constraints~\eqref{LPs:R},~\eqref{LPs:star} and~\eqref{LPs:R-d} hold, and the cost of the solution is also at most $(1+\delta/2)^q{U'} = \tilde U'$ from Property~\ref{property:prep-good-restated}.

In the pseudo-approximation setting, the lemma holds trivially since \eqref{LP:basic} is a valid LP relaxation to the instance $\calI$.
\end{proof}


\subsection{Eliminating the $x_{ij}$ variables}
	\label{subsec:splitting}
Suppose we have an optimal solution $(x, y)$ to \eqref{LP:strong}, such that for all $i,j$, either $x_{ij} = y_i$ or $x_{ij} = 0$, then we can easily eliminate the $x_{ij}$ variables and deal with an LP purely over the $y$ variables. Indeed, for the standard $k$-median problem, such a step, which is the starting point for many primal-based rounding algorithms, is easy by simply splitting each $y_i$ and making collocated copies of facilities in $F$. In our case, we need to take extra care to handle Constraint~\eqref{LPs:star}.

\begin{lemma} \label{properties:ystar-after-splitting-lemma}
By adding collocated facilities to $F$, we can efficiently obtain a vector $y^* \in [0, 1]^F$ and a set $F_j \subseteq \Ball_F(j, R_j)$ for every $j \in C$, with the following properties.
\begin{enumerate}[topsep=3pt,itemsep=0pt,label=(\ref{properties:ystar-after-splitting-lemma}\alph*),leftmargin=*]
	\item For each client $j$, we have $y^*(F_j) \leq 1$. \label{property:y-star-bound}
    \item The total fractionally open facilities satisfies $y^*(F) \leq k$. \label{property:y-k-bound}
	\item The total client coverage satisfies $\sum_{j\in C} y^*(F_j) \geq m$. \label{property:y-star-coverage}
    \item The solution cost is bounded: $\sum_{j \in C} \sum_{i \in F_j} d^q(i,j) y^*_i \leq \tilde U'$. \label{property:y-cost}
    \item For each $i \in S_0$, we have $\sum_{i'~\text{collocated with}~i} y^*_{i'} = 1$. \label{property:y-s0}
	\item For every facility $i$ that is not collocated with any facility in $S_0$, the ``star-cost'' of $i$ satisfies $\sum_{j \in C: i \in F_j}$ $d^q(i, j) \leq 2\rho \tilde U$. \label{property:y-star-star}
\end{enumerate}
\end{lemma}

	\begin{proof}
		Let us denote the solution obtained by solving ~\eqref{LP:strong} as $(x, y)$.
		To avoid notational confusion, we first create a copy $F'$ of $F$ and transfer all the $y$ values from $F$ to $F'$, i.e., for each facility $i \in F$, there is a facility $i' \in F'$ collocated with $i$ and we set $y^*_{i'} = y_i$. Our final vector $y^*$ (and all $F_j$'s) will be entirely supported only in $F'$. We shall create a set $F_j \subseteq F'$ for each $j$, where initially we have $F_j = \emptyset$ for every $j$. The star cost of a facility $i' \in F'$ is simply $\sum_{j \in C: i' \in F_j} d^q(i', j)$. During the following procedure, we may split a facility $i' \in F'$ into two facilities $i'_1$ and $i'_2$ collocated with $i'$ with $y^*_{i'_1} + y^*_{i'_2} = y^*_{i'}$. Then, we update $F' \gets F' \setminus \set{i'} \cup \set{i'_1, i'_2}$ and for each $j$ such that $i' \in F_j$, we shall update $F_j \gets F_j \setminus \set{i'} \cup \set{i'_1, i'_2}$.
		
		
		
		For every facility $i \in F$, for every client $j \in C$, with $x_{i, j} > 0$, we apply the following procedure. We order the facilities in $F'$ collocated with $i$ in non-decreasing order of their current star costs.  We choose the first $o$ facilities in this sequence whose total $y^*$ value is exactly $x_{i, j}$; some facility may need to be split in order to guarantee this. We then add these facilities to $F_j$.
		
		
		Properties~\ref{property:y-star-bound} to~\ref{property:y-s0} are easy to see since the new solution given by $y^*$ and $(F_j)_{j \in C}$ is the same as the solution $(x, y)$ up to splitting of facilities. It remains to prove Property~\ref{property:y-star-star}; this is trivial for the pseudo-approximation setting as $\tilde U = \infty$. 
		We can view the above procedure for any fixed $i \in F$ as a greedy algorithm for makespan minimization on identical machines. Let $M$ be a large integer such that all considered fractional values are multiplies of $1/M$. There are $My_i$ machines, and for every $j \in C$, there are $Mx_{i, j}$ jobs of size $d^q(i,j)$ corresponding to each client $j$.  There is an extra constraint that all the jobs corresponding to $j$ must be scheduled on different machines.  The above algorithm is equivalent to the following: for each $j \in C$, we choose the $Mx_{i, j}$ machines with the minimum load and assign one job correspondent to $j$ to each machine. (This algorithm may run in exponential time; but it is only for analysis purpose.)
			
			First, note that size of any job is at most $\rho \tilde U$ from \eqref{LPs:R-d}.  Thus, the highest load over all machines is at most the minimal one plus $\rho \tilde U$.  The minimal load over all machines is at most the average load, which is equal to $\left( \sum_{j}d^q(i,j) \cdot \frac{ Mx_{ij}}{My_i} \right) \leq \rho\tilde U$, by \eqref{LPs:star}. Thus, the maximum load is at most $2\rho \tilde U$.  Now, we redefine $F$ as $F'$ and finish the proof of the Lemma.
		%
		%
		%
	\end{proof}


\subsection{Random Discretization of Distances} \label{sec:discretize}
Our next step in the rounding algorithm is done in order to optimize our final approximation factor. To this end, if $q=1$, let $\tau = \arg \min \frac{3\tau-1}{\ln\tau} \approx 2.3603$, and then $\alpha_1 = \frac{3\tau-1}{\ln\tau} < 7.081$; if $q=2$, let $\tau = \arg \min \frac{(\tau + 1)(3\tau-1)^2}{2(\tau-1)\ln \tau} \approx 2.24434$ and then $\alpha_2 = \frac{(\tau + 1)(3\tau-1)^2}{2(\tau-1)\ln \tau} \leq 53.002$. We choose a random offset $a \sim [1, \tau)$ such that $\ln a$ is uniformly distributed in $[0, \ln \tau)$. Then, define $D_{-2} = -1, D_{-1} = 0$, and $D_\ell = a \tau^\ell$ for every integer $\ell \geq 0$. Now, we \emph{increase} each distance $d$ to the smallest value $d' \geq d$ belonging to the set $\set{D_{-1}, D_0, D_1, D_2, \cdots,}$. Formally, for every $i \in F,j \in C$, let $d'(i,j) = D_\ell$, where $\ell$ is the minimum integer such that $d(i, j)\leq D_\ell$. Note that the distances $d'(i,j)$ may not satisfy the triangle inequality anymore, but nevertheless this discretization will serve useful in our final analysis for optimizing the approximation factor.

\begin{lemma}
\label{lemma:discrete}
	For all $i,j$, we have $d'(i,j) \geq d(i,j)$ and $\E_a [{d'}^q(i,j)] = \frac{\tau^q-1}{q\ln \tau} d^q (i,j)$.
\end{lemma}

	\begin{proof}
		We can assume $d(i, j) \geq 1$, since the case for $d(i, j) = 0$ is trivial. Let $\beta = \log_{\tau}{a}$. $\beta$ is distibuted uniformly in $[0,1)$. Let $d(i,j) = \tau^{\ell + p}$ where $\ell$ is an integer and $0\leq p<1$. When $\beta$ is less than $p$, $d(i,j)$ is rounded to $d'(i,j)=\tau^{\ell+1+\beta}$. If $\beta$ is at least $p$, then $d'(i,j)=\tau^{\ell+\beta}$.
\ifdefined\fullversion
			\begin{align*}
				\E_a\left[{d'}^q(i,j)\right] &= \int_{\beta=0}^{p}\tau^{q(\ell+1+\beta)}+\int_{\beta=p}^{1}\tau^{q(\ell+\beta)}
				 =  \tfrac{\tau^{q(\ell+1+\beta)}}{q\ln{\tau}}\Big|_{0}^{p} + \tfrac{\tau^{q(\ell + \beta)}}{q\ln{\tau}}\Big|_{p}^{1}\\
				 &=  \tfrac{\tau^{q(\ell+1+p)}-\tau^{q(\ell+1)}}{q\ln{\tau}}+ \tfrac{\tau^{q(\ell+1)}-\tau^{q(\ell+p)}}{q\ln{\tau}}
				=  \tfrac{\tau^{q(\ell+1+p)}-\tau^{q(\ell+p)}}{q\ln{\tau}}
				 =  \tau^{q(\ell+p)}\tfrac{\tau^q-1}{q\ln{\tau}}
				 = \tfrac{\tau^q-1}{q\ln{\tau}}d^q(i,j).\hfill \qedhere
			\end{align*}
\else
			\begin{align*}
				\E_a\left[{d'}^q(i,j)\right] &= \int_{\beta=0}^{p}\tau^{q(\ell+1+\beta)}+\int_{\beta=p}^{1}\tau^{q(\ell+\beta)} \\ \allowdisplaybreaks
				 &=  \tfrac{\tau^{q(\ell+1+\beta)}}{q\ln{\tau}}\Big|_{0}^{p} + \tfrac{\tau^{q(\ell + \beta)}}{q\ln{\tau}}\Big|_{p}^{1}\\
				 &=  \tfrac{\tau^{q(\ell+1+p)}-\tau^{q(\ell+1)}}{q\ln{\tau}}+ \tfrac{\tau^{q(\ell+1)}-\tau^{q(\ell+p)}}{q\ln{\tau}}\\ 
				&=  \tfrac{\tau^{q(\ell+1+p)}-\tau^{q(\ell+p)}}{q\ln{\tau}}
				 =  \tau^{q(\ell+p)}\tfrac{\tau^q-1}{q\ln{\tau}}
				 = \tfrac{\tau^q-1}{q\ln{\tau}}d^q(i,j).\hfill \qedhere
			\end{align*}
\fi		
	\end{proof}


\subsection{Auxiliary LP and Iterative Rounding}

The crucial ingredient of our rounding procedure is the following \emph{auxiliary LP} \eqref{LP:iter} which we iteratively update and solve. At each step, we maintain a partition of $C$ into $\cfull$, the set of clients which need to be fully connected,
and $\cpart$, the set of partially connected clients. For each client $j$, we also maintain a set $F_j$ of allowed connections, a \emph{radius-level} $\ell_j$ for each client $j$ such that $D_{\ell_j}$ is the maximum connection cost for $j$, and a set $B_j$, called the \emph{inner-ball} of $j$ which only includes facilities from $F_j$ at distance at most $D_{\ell_j-1}$ from $j$. Finally, we also maintain a subset \cstar~of full clients such that their $F_j$ balls are disjoint, and also such that, every other full client $j \notin \cstar$~is ``close'' (within $O(1) D_{\ell_j}$) to the set \cstar. A small technicality is that for each facility $i \in S_0$, we also view $i$ as a \emph{virtual client} that we shall include in $C^*$ at the beginning. This is done in order to ensure each $i \in S_0$ will be open eventually: since we split the facilities, $y^*$ obtained from~\Cref{properties:ystar-after-splitting-lemma} only has $\sum_{i'\text{ collocated with }i}y^*_{i'} = 1$ for every $i \in S_0$. Our algorithm will operate in a way that these virtual clients will never get removed from $C^*$, using which we will show that we open $S_0$ in the end.

Initially, $\cfull \gets \emptyset$, $\cpart\gets C$ and $\cstar \gets S_0$. For each client $j \in C$, the initial set $F_j$ is the one computed in Lemma~\ref{properties:ystar-after-splitting-lemma}, and we define $\ell_j$ to be the integer such that $D_{\ell_j} = \max_{i \in F_j}d'(i, j)$.  For a virtual client $j \in S_0$, let $F_j$ be the set of facilities collocated with $j$ and thus $y^*(F_j) = 1$, and define $\ell_j = -1$ (thus $D_{\ell_j} = 0$). We remark that as we proceed with the rounding algorithm, the $F_j$'s and $\ell_j$'s will be updated but $F_j$ will only \emph{shrink} and $\ell_j$ can only decrease. Crucially, at every step the set $F_j$ will always be contained in $\Ball_F(j,D_{\ell_j})$.
\vspace*{-15pt}

\ifdefined\fullversion
	\begin{align}
	  \min \qquad \sum_{j \in \cpart} \sum_{i \in F_j} {d'}^q(i, j)y_i
	  +  \sum_{j \in \cfull}\left( \sum_{i \in B_j} {d'}^q(i, j)y_i + (1 - y(B_j)) D^q_{\ell_j} \right)  \qquad \textrm {s.t.}  \tag{$\text{LP}_{\text{iter}}$} \label{LP:iter}
	\end{align}\vspace*{-20pt}
	
	\noindent\begin{minipage}{0.4\textwidth}
		\begin{alignat}{2}
	  y(F) &\leq \, k   \label{LPiter:1} \\[3pt]
			y(F_j) &= \, 1 &\qquad &\forall j \in \cstar \label{LPiter:y-full} \\[12pt]
			y(B_j) &\leq \, 1 &\qquad &\forall j \in \cfull  \label{LPiter:y-B}
		\end{alignat}
	\end{minipage}
	\begin{minipage}{0.6\textwidth}
		\begin{alignat}{2}
			y(F_j) &\leq \, 1 &\qquad &\forall j \in \cpart \label{LPiter:y-part}\\
	   		|\cfull |+\sum_{j\in \cpart}y(F_j) &\geq \, m \label{LPiter:coverage} \\
	   		y_i &\in  [0, 1] &\qquad &\forall i \in F \label{LPiter:6}
		\end{alignat}
	\end{minipage}\bigskip
\else
	\begin{multline}
	  \min \qquad \sum_{j \in \cpart} \sum_{i \in F_j} {d'}^q(i, j)y_i \\
	  +  \sum_{j \in \cfull}\left( \sum_{i \in B_j} {d'}^q(i, j)y_i + (1 - y(B_j)) D^q_{\ell_j} \right)  \qquad \textrm {s.t.}  \tag{$\text{LP}_{\text{iter}}$} \label{LP:iter}
	\end{multline}\vspace*{-20pt}

	\begin{align}
		y(F_j) &= 1 &  &\forall j \in \cstar &\label{LPiter:y-full}\\
		y(F) &\leq k  & &  &\label{LPiter:1}\\
		y(B_j) &\leq 1 & & \forall j \in \cfull &\label{LPiter:y-B}\\
		y(F_j) &\leq 1 & &\forall j \in \cpart &\label{LPiter:y-part}\\
		 |\cfull |+\sum_{j\in \cpart}y(F_j) &\geq m & & &\label{LPiter:coverage}\\
		  y_i &\in  [0, 1] & &\forall i \in F &\label{LPiter:6}
	\end{align}
\fi

In the above LP, for a client $j \in \cpart$, the quantity $y(F_j)$ denotes the extent that $j$ is connected, hence the constraint~\eqref{LPiter:y-part}. Constraint~\eqref{LPiter:coverage} enforces that the total number of covered clients (full plus partial) is at least $m$. Then, constraint~\eqref{LPiter:y-full} ensures that clients in $\cstar$ are covered fully.
Since we don't enforce~\eqref{LPiter:y-full} for full clients that are not in $\cstar$, we make sure that such client are close to $\cstar$.

Now, we describe the objective function of \eqref{LP:iter}. Notice that we use $d'$ instead of $d$ in the objective function. For any client $j \in \cpart$, $y(F_j)$ is precisely the extent to which $j$ is connected, and so, we can simply use $\sum_{i \in F_j}{d'}^q(i, j) y_i$ to denote the connection cost of $j$. For a client $j \in \cfull$, $j$ is required to be fully connected but we may not have $y(F_j) = 1$. Hence, $\sum_{i \in F_j} {d'}^q(i,j) y_i$ is no longer a faithful representation of its connection cost, and hence we need to express the cost in a different manner. To this end, for every $j \in \cfull$, we require $y(B_j) \leq 1$ in~\eqref{LPiter:y-B}. We guarantee that $B_j$ is always $\set{i \in F_j: d'(i, j)\leq D_{\ell_j - 1}}$. Then, for connections between $j$ and any facility $i \in B_j$, we use the rounded distance $d'(i, j)$ in the objective function. For the remaining $1-y(B_j)$ fractional connection of $j$, we use the term $D_{\ell_j}$. This gives us the objective function, and in turn completes the description of~\eqref{LP:iter}.
\begin{lemma}
	\label{lem:initial-y-star-good}
The $y^*$ computed in Section~\ref{subsec:splitting} is a feasible solution to \eqref{LP:iter}. Moreover, the expected value of the objective function (over the random choice of $a$ defined in~\Cref{sec:discretize}) is at most $\frac{\tau^q-1}{q\ln \tau}\tilde U'$.
\end{lemma}
\begin{proof}
	Initially we have $\cstar=S_0$, $\cfull=\emptyset$ and $\cpart=C$. The feasibility of $y^*$ is guaranteed by Properties~\ref{property:y-star-bound}-\ref{property:y-star-coverage}, and the fact that for every virtual client $j \in S_0$, we have $y^*(F_j) = 1$. Now, since $\cpart = C$ and $\cfull = \emptyset$, the objective value of $y^*$ w.r.t \eqref{LP:iter} is $\sum_{j \in C, i \in F_j}d'(i, j)y^*_i$. Now, from Property~\ref{property:y-cost}, we have that $\sum_{j \in C, i \in F_j}d(i, j)y^*_i \leq \tilde U'$. Combining this with Lemma~\ref{lemma:discrete} bounds the objective value.
\end{proof}


\begin{algorithm}[h]
	\caption{Iterative Rounding Algorithm} \label{alg:iter}
	\begin{itemize}[leftmargin=*]
		\item 	\textbf{Input}: $\mathcal{I} = (F, C, d, k, m, S_0), \rho,\delta \in (0, 1/2), U \geq 0, y^* \in [0, 1]^F, (F_j)_{j \in C \cup S_0}$ and $(\ell_j)_{j \in C \cup S_0}$
		\item \textbf{Output}: a new solution $y^*$ which is almost-integral
	\end{itemize}\vspace*{-5pt}
	
	\rule{\linewidth}{0.5pt}
	\begin{algorithmic}[1]
		\State $\cfull \gets \emptyset, \cpart \gets C, \cstar \gets S_0$
		\While{true} \Comment{(main loop)}
			\State find an optimum vertex point solution $y^*$ to~\eqref{LP:iter} \label{state:iter-solve-LP}
			\If{there exists some $j \in C_{\parti}$ such that $y^*(F_j) = 1$}
				\State $C_{\parti} \gets C_{\parti} \setminus \set{j}$,  $C_{\full} \gets C_{\full} \cup \set{j}$, \ifdefined\fullversion\else\newline\fi$B_j \gets \set{i \in F_j: d'(i, j) \leq D_{\ell_j - 1}}$
                $\textsf{update-}C^*(j)$ \label{state:iter-move-j-to-full}
			\ElsIf{there exists $j \in C_\full$ such that $y^*(B_j) = 1$}
				\State $\ell_j \gets \ell_j-1$, $F_j \gets B_j$, \ifdefined\fullversion\else\newline\fi $B_j \gets \set{i \in F_j: d'(i, j) \leq D_{\ell_j-1}}$
                $\textsf{update-}C^*(j)$ \label{state:iter-decrease-lj}
			\Else
				\State break
			\EndIf
		\EndWhile
		\State \Return $y^*$
	\end{algorithmic}
	\vspace*{-5pt}
	\rule{\linewidth}{0.5pt}
	
	$\textsf{update-}C^*(j)$:
	\begin{algorithmic}[1]
		\If {there exists no $j' \in C^*$ with $\ell_{j'} \leq \ell_j$ and $F_j \cap F_{j'} \neq \emptyset$}
			\State remove from $C^*$ all $j'$ such that $F_j \cap F_{j'} \neq \emptyset$
			\State $C^* \gets C^* \cup \set{j}$
		\EndIf
	\end{algorithmic}
\end{algorithm}

We can now describe our iterative rounding algorithm, formally stated in Algorithm~\ref{alg:iter}. In each iteration, we solve the LP to obtain a vertex solution $y^*$ in Step~\ref{state:iter-solve-LP}. If~\eqref{LPiter:y-part} is tight for some partial client $j \in \cpart$, then we update $\cfull \gets \cfull \cup \{j\}$ and remove $j$ from $\cpart$. We also update $\cstar$ to ensure that there is a facility in $\cstar$ that is close to all full clients; this is done in the procedure $\textsf{update-}C^*$. Likewise if \eqref{LPiter:y-B} is tight for some client $j \in \cfull$, then we decrease $\ell_j$ by 1, update $F_j \gets B_j$ and $B_j$ to be an even smaller set. We again call $\textsf{update-}C^*(j)$ to ensure \cstar~is close to all full clients --- this is needed since we decreased $\ell_j$ and the definition of ``close'' becomes more strict. If neither of the constraints~\eqref{LPiter:y-B} and \eqref{LPiter:y-part} are tight for $y^*$, then we show that $y^*$ is almost integral and we return $y^*$.

%

\subsection{Analysis}
Our proof proceeds as follows: we first show that, if the vertex point $y^*$ computed satisfies all constraints~\eqref{LPiter:y-B} and~\eqref{LPiter:y-part} with strict inequality, then it is almost integral. Next we show that the objective value of the solutions computed is non-increasing as we iterate. Then, we show that all full clients are close to \cstar, using which we conclude that there is one unit of fractional facilities opened within distance $O(1) D_{\ell_j}$ for all $j \in \cfull$. This suggests that the objective function of \eqref{LP:iter} captures the actual cost well. Finally, we show that all virtual clients in $S_0$ will always be in $C^*$, implying that there will be an integrally open facility at each location in $S_0$. Combining all these leads to an almost-integral solution with bounded cost.

We begin by describing some simple invariants that are maintained in the algorithm; throughout this section, we assume that every step indicated by a number in Algorithm~\ref{alg:iter} is atomic.

\begin{claim} \label{cl:invts}
	After each step of the algorithm, the following hold.
	\begin{enumerate}[topsep=3pt,itemsep=0pt, label=\arabic*]
		\item The sets $\cfull$ and $\cpart$ form a partition of $C$, $S_0 \subseteq C^*$ and $\cstar \setminus S_0 \subseteq \cfull$.
		\item The sets $\set{F_j:j \in C^*}$ are mutually disjoint.
		\item If $j \in \cfull$, then $B_j = \set{i \in F_j: d'(i, j) \leq D_{\ell_j - 1}}$.
		\item For every $j \in C$ and $i \in F_j$, we have $d'(i, j) \leq D_{\ell_j}$.
        \item For every $j$, $D_{\ell_j} \leq \tau R_j$.
		\item For every $j$, $\ell_j \geq -1$.
	\end{enumerate}
\end{claim}

	\begin{proof}
		$\cfull$ and $\cpart$ form a partition of $C$ since we only \emph{move} clients from $\cpart$ to $\cfull$.  Moreover, virtual clients in $S_0$ will never be removed from $C^*$ since each such client $j$ has $\ell_j = -1$. Finally, we call $\textsf{update-}C^*(j)$ only if $j$ is already in $\cfull$ and we only add $j$ to $C^*$ in the procedure. Thus, Property 1 holds. Also, before we add $j$ to $C^*$ in $\textsf{update-}C^*(j)$, we removed all $j'$ from $C^*$ such that $F_j \cap F_{j'} \neq \emptyset$, and so Property 2 holds.  Property 3 holds since every time we move $j$ from $\cpart$ to $\cfull$, or update $F_j$ for some $j \in \cfull$, we define $B_j = \set{i \in F_j: d'(i, j) \leq D_{\ell_j - 1}}$. Property 4 holds at the beginning of the algorithm, and before we change $F_j$ to $B_j$ in Step~\ref{state:iter-decrease-lj}, we have $d'(i, j) \leq D_{\ell_j - 1}$ for every $i \in B_j$. Property 5 holds initially because initially $D_{\ell_j} = \max_{i \in F_j} d'(i,j) \leq \tau \max_{i \in F_j} d(i,j) \leq \tau R_j$, and $D_{\ell_j}$ is non-increasing over time. Property 6 holds since if $\ell_j = -1$ for some $j \in \cfull$ then $B_j$ is empty thus $y^*(B_j)$ will never become 1.
	\end{proof}

We now proceed to the analysis of the algorithm, by proving the following lemmas.

\begin{lemma} \label{lem:2fractional}
	If after Step~\ref{state:iter-solve-LP} in Algorithm~\ref{alg:iter}, none of the constraints~\eqref{LPiter:y-B} and~\eqref{LPiter:y-part} are tight for $y^*$, then there are at most $2$ strictly fractional $y^*$ variables.
\end{lemma}

	\begin{proof}
		Assume $y^*$ has $n' \geq 3$ fractional values. We pick $|F|$ independent tight inequalities in \eqref{LP:iter} that defines $y^*$. Without loss of generality, we can assume if $y^*_i \in \set{0, 1}$, then the constraint \eqref{LPiter:6} for $i$ is picked. Thus, there are exactly $|F|-n'$ tight constraints from among \eqref{LPiter:6} and $n'$ tight constraints from among \eqref{LPiter:1}, \eqref{LPiter:y-full} or \eqref{LPiter:coverage}. Each tight constraint of form~\eqref{LPiter:y-full} should contain at least 2 fractional variables, and the tight constraint \eqref{LPiter:1} should contain 2 fractional variables, that are not contained in any tight constraint of form~\eqref{LPiter:y-full}. Thus, the number of tight constraints of form \eqref{LPiter:1}, \eqref{LPiter:y-full} or \eqref{LPiter:coverage} is at most $n'/2 + 1$. This is less than $n'$ if $n' \geq 3$. Thus, $n' \leq 2$.
	\end{proof}


\begin{lemma}
	\label{lem:y-valid}
	After any step of the algorithm, the maintained solution $y^*$ is feasible to \eqref{LP:iter}.
\end{lemma}

	\begin{proof}
		We only need to consider the time point after we run Step~\ref{state:iter-move-j-to-full} or \ref{state:iter-decrease-lj} during the algorithm. In Step~\ref{state:iter-move-j-to-full}, we move $j$ from $\cpart$ to $\cfull$ and define $B_j$. We were guaranteed that $y^*(F_j) = 1$ and thus after the step we have $y^*(B_j) \leq y^*(F_j) \leq 1$. Before Step~\ref{state:iter-decrease-lj}, we were guaranteed $y^*(B_j) = 1$. Thus, after the step, we have $y^*(F_j) = 1$ and $y^*(B_j) \leq 1$.
	\end{proof}

\begin{lemma} \label{lem:cost}
	In every iteration, the~\eqref{LP:iter} objective value of the solution $y^*$ is non-increasing.
\end{lemma}

	\begin{proof}
		In Statement~\ref{state:iter-solve-LP}, we resolve the LP and thus the value of $y^*$ can only go down (due to Lemma~\ref{lem:y-valid}). Thus, we only need to consider the situation in which we change the objective function, which happens in Step~\ref{state:iter-move-j-to-full} and \ref{state:iter-decrease-lj}. Indeed, we show that the objective value is not affected by these operations.
		\begin{enumerate}	
			\item Consider the situation where a client $j$ is moved to $\cfull$ from $\cpart$ in Statement~\ref{state:iter-move-j-to-full}. The cost that $j$ is contributing in the old LP is
\ifdefined\fullversion
			 \begin{align*}
				 \sum_{i\in F_j}{d'}^q(i,j)y^*_i
				 &=\sum_{i\in F_j:{d'}(i,j)\leq D_{\ell_j-1}}{d'}^q(i,j)y^*_i+\sum_{i\in F_j:{d'}(i,j) = D_{\ell_j}}{d'}^q(i,j)y^*_i \\
				 &= \sum_{i\in B_j}{d'}^q(i,j)y^*_i+(1-y^*(B_j))D_{\ell_j}^q,
			\end{align*}
\else
			 \begin{align*}
				 &\quad \sum_{i\in F_j}{d'}^q(i,j)y^*_i \\
				 &=\sum_{i\in F_j:{d'}(i,j)\leq D_{\ell_j-1}}{d'}^q(i,j)y^*_i+\sum_{i\in F_j:{d'}(i,j) = D_{\ell_j}}{d'}^q(i,j)y^*_i \\
				 &= \sum_{i\in B_j}{d'}^q(i,j)y^*_i+(1-y^*(B_j))D_{\ell_j}^q,
			\end{align*}
\fi
			   which is exactly the contribution of $j$ to the new objective function.
			\item Then, consider Statement~\ref{state:iter-decrease-lj}. To avoid confusion, let $\ell_j$, $F_j$ and $B_j$ correspond to the values before executing the statement, and $\ell'_j$, $F'_j$ and $B'_j$ correspond to the values after executing the statement.  The contribution of $j$ to the old function is
\ifdefined\fullversion
			\begin{align*}
				&\qquad \sum_{i \in B_j}{d'}^q(i, j) y^*_i + \left(1-y(B_j)\right)D^q_{\ell_j}
				= \sum_{i \in B_j}{d'}^q(i, j) y^*_i = \sum_{i \in F'_j}{d'}^q(i, j) y^*_i\\
				&= \sum_{i \in F'_j: d'(i, j) \leq D_{\ell'_j-1}} {d'}^q(i, j)y^*_i + \sum_{i \in F'_j: d'(i, j) = D_{\ell'_j}} {d'}^q(i, j) y^*_i
				= \sum_{i \in B'_j} {d'}^q(i, j)y^*_i + \left(1-y(B'_j)\right)D^q_{\ell'_j},
			\end{align*}
\else
			\begin{align*}
				&\qquad \sum_{i \in B_j}{d'}^q(i, j) y^*_i + \left(1-y(B_j)\right)D^q_{\ell_j}\\
				&= \sum_{i \in B_j}{d'}^q(i, j) y^*_i = \sum_{i \in F'_j}{d'}^q(i, j) y^*_i \\
				&= \sum_{i \in F'_j: d'(i, j) \leq D_{\ell'_j-1}} {d'}^q(i, j)y^*_i + \sum_{i \in F'_j: d'(i, j) = D_{\ell'_j}} {d'}^q(i, j) y^*_i \\
				&= \sum_{i \in B'_j} {d'}^q(i, j)y^*_i + \left(1-y(B'_j)\right)D^q_{\ell'_j},
			\end{align*}
\fi
			which is the contribution of $j$ to the new function. \qedhere
		\end{enumerate}
	\end{proof}

\begin{lemma}
\label{lem:5r}
	At the conclusion of the algorithm, for every $j\in \cfull$, there exists at least $1$ unit of open facilities within distance $\frac{3\tau-1}{\tau-1}D_{\ell_{j}}$ from $j$. Formally, $\sum_{i : d(i,j) \leq \frac{3\tau-1}{\tau-1}D_{\ell_{j}}} y^*_i \geq 1$.
\end{lemma}

	\begin{proof}
		The following claims will serve useful in the proof of the lemma.
		
		\begin{claim} \label{cl:dj1j2}
		If at any stage, $F_j \cap F_{j'} \neq \emptyset$, then we have that $d(j,j') \leq D_{\ell_j}+ D_{\ell_{j'}}$.
		\end{claim}
		
		\begin{claim} \label{cl:evict}
			Consider a client $j$ which is added to $\cstar$ at Step 3 of $\textsf{update-}C^*(j)$ when its radius-level $\ell_j$ was some $\ell \geq -1$. Then, there exists at least $1$ unit of open facilities in the final solution $y^*$ within distance $\frac{\tau+1}{\tau-1}D_\ell$ from $j$.
		\end{claim}
	
		\begin{proof}
			The proof is by induction on radius-levels. Indeed, clients with radius-level $-1$ can never be removed from $\cstar$, so if such clients are added to $\cstar$, then they remain in $\cstar$ at the end, and so the claim is true. So suppose the claim is true up to some radius-level $\ell$-1, and consider a client $j$ which got added to $\cstar$ with radius-level $\ell$ at some time in Step 3 of $\textsf{update-}C^*(j)$. Then, either it remains in $\cstar$ till the end, in which case we are done (because its $\ell_j$ is non-increasing over time), or another client $j'$ was subsequently added to $\cstar$ such that $\ell_{j'} < \ell_j$ and $F_j \cap F_{j'} \neq \emptyset$. Here, $\ell_j$ and $\ell_{j'}$ are both the radius-levels of $j$ and $j'$ at the time $j$ was removed by $j'$. Since these values are non-increasing over time, we get that $\ell_{j'} < \ell_j \leq \ell$. Now, in this latter case, note that by the induction hypothesis there exists one unit of facility in the final solution $y^*$ within distance $\frac{\tau+1}{\tau-1}D_{\ell_{j'}} \leq \frac{\tau+1}{\tau-1} \frac{D_{\ell_{j}}}{\tau} \leq \frac{\tau+1}{\tau-1} \frac{D_{\ell}}{\tau}$ from $j'$. Moreover, we know $d(j,j') \leq D_{\ell} + D_{\ell}/\tau$ from Claim~\ref{cl:dj1j2}. The proof then follows from the triangle inequality and $\frac{\tau+1}{\tau-1}\frac{D_\ell}{\tau} + D_\ell + \frac{D_\ell}{\tau} = \frac{\tau+1}{\tau-1}D_\ell$.
		\end{proof}

		Now the proof of~\Cref{lem:5r} is simple: consider a client $j$ which belongs to $\cfull$ at the end of the algorithm, and consider the last time when $\textsf{update-}C^*(j)$ is invoked, and let $\ell_j$ denote its radius-level at this time. Note that the value $\ell_j$ has not changed subsequently.
		
		At this point, there are two cases depending on whether $j$ was added to $\cstar$ or not. If it was added, we are done by Claim~\ref{cl:evict}. In case it was not added, then there must be a client $j'$ which belonged to $\cstar$ such that $F_{j'} \cap F_j \neq \emptyset$ and $\ell_{j'} \leq \ell_j$. Then, note that $d(j,j') \leq 2 D_{\ell_j}$ from Claim~\ref{cl:dj1j2}, and moreover, by Claim~\ref{cl:evict}, there is at least $1$ unit of open facilities in the final solution $y^*$ within distance $\frac{\tau+1}{\tau-1}D_{\ell_{j'}} \leq \frac{\tau+1}{\tau-1}D_{\ell_{j}}$ from $j'$. We can then complete the proof by applying triangle inequality.
	\end{proof}

With all the above lemmas, we can conclude this section with our main theorem.

\begin{theorem} \label{thm:rounding}
	When the algorithm terminates, it returns a valid solution $y^* \in [0, 1]^F$ to \eqref{LP:basic} with at most two fractional values. $y^*_i = 1$ for every $i \in S_0$. Moreover, the cost of $y^*$ w.r.t \eqref{LP:basic}, in expectation over the randomness of $a$ is at most $\alpha_q \tilde U'$.
\end{theorem}

	\begin{proof}
		Firstly, at the culmination of our iterative algorithm, note that none of the constraints~\eqref{LPiter:y-B} and~\eqref{LPiter:y-part} are tight for the final $y^*$, otherwise our algorithm would not have terminated. As a result, it follows from~\Cref{lem:2fractional} that there are at most two strictly fractional $y^*$ variables.
		
		As for bounding the cost, we show this by exhibiting a setting of $x_{ij}$ for the given $y^*$ with the desired bound on cost.  To this end, for clients in $\cpart$, set $x_{ij} = y^*_i$ for all $i \in F_j$. For clients in $\cfull$, set $x_{ij} = y^*_i$ if $i \in B_j$, and set the remaining $1 - y^*(B_j)$ fractional amount arbitrarily among open facilities within distance $\frac{3\tau-1}{\tau-1} D_{\ell_j}$ from $j$ ---~\Cref{lem:5r} guarantees the existence of such open facilities.
		
		 We now verify that such a solution satisfies the constraints of~\eqref{LP:basic}. Indeed, it is easy to see that $\sum_{i} y^*_i \leq k$, and also that $x_{i,j} \leq y^*_i$. As for the coverage constraint, note that for full clients $j$, we have $\sum_{i} x_{i,j} = 1$ and for partial clients $j$, we have $\sum_{i} x_{i,j} = \sum_{i \in F_j} y^*_i = y^*(F_j) \leq 1$. Finally the total coverage $\sum_{j,i} x_{i,j} = |\cfull| + \sum_{j \in \cpart}y^*(F_j) \geq m$ since $y^*$ satisfies~\eqref{LPiter:coverage}.  Also, every virtual client $j \in S_0$ is always in $C^*$ and we always have $F_j$ is the set of facilities collocated with $j$. Eventually we have $y^*(F_j) = 1$. W.l.o.g we can assume $y^*_i = 1$ for every $i \in S_0$.
		
		To complete the proof, we compute the cost of the solution we have constructed. Indeed, the contribution of a partial client $j$ in our solution is at most their contribution to the objective value of $y^*$ for~\eqref{LP:iter}. This is because $\sum_{i} x_{i,j} d(i,j) \leq \sum_{i} x_{i,j} d'(i,j) = \sum_{i \in F_j} d'(i,j) y^*_i$. Now consider a full client $j \in \cfull$. By our setting we have $x_{i,j} = y^*_i$ for $i \in B_j$ and hence $\sum_{i \in B_j} x_{i,j} d(i,j) \leq \sum_{i \in B_j} x_{i,j} d'(i,j) = \sum_{i \in B_j} y^*_i d'(i,j)$. Hence, the contribution to the objective value of assignment within $B_j$ for full clients has a one-to-one correspondence with their contribution in the auxiliary LP~\eqref{LP:iter}. It remains to bound the connection cost to facilities outside $B_j$. For bounding this, note that the extend of outside connections for each client $j$ is exactly $1 - y^*(B_j)$, and from~\Cref{lem:5r}, we get that all such connections are at a distance at most $\frac{3\tau-1}{\tau-1} D_{\ell_j}$ from $j$, incurring a cost of $\frac{(3\tau-1)^q}{(\tau-1)^q} D^q_{\ell_j}$. In contrast, these connections contribute a total of $D^q_{\ell_j} (1 - y^*(B_j))$  to the objective value of $\eqref{LP:iter}$. Thus, the cost of our solution w.r.t \eqref{LP:basic} is at most $\frac{(3\tau -1)^q}{(\tau - 1)^q}$ times the cost of $y^*$ w.r.t \eqref{LP:iter}, which in expectation is at most $\frac{\tau^q-1}{q\ln \tau}$ times $\tilde U'$, by Lemma~\ref{lem:initial-y-star-good} and \ref{lem:cost}. Hence, the cost of our constructed almost-integral solution w.r.t. \eqref{LP:basic} is at most $\frac{(3\tau -1)^q}{(\tau - 1)^q}\cdot \frac{\tau^q-1}{q\ln \tau} \cdot \tilde U'$ in expectation. By the way we choose $\tau$, this is exactly $\alpha_q \tilde U'$.
	\end{proof}
{\em Pseuso-approximation Setting:} for this setting, we simply change the two possible fractionally open facilities to be integrally open, and hence we can obtain an $\alpha_q$-approximation for \kmwo/\kmeanswo with $k+1$ open facilities. 

%% file: fractional.tex
\section{Opening exactly $k$ facilities}
\label{sec:open-k}
In this section, we show how to convert the almost-integral $y^*$ computed by the iterative rounding~\Cref{alg:iter} into a fully integral solution with bounded cost. Indeed, this last step of converting an almost-integral solution into a fully-integral one is why we needed the three steps of pre-processing and the stronger LP relaxation in the first place.
From~\Cref{thm:rounding}, we know that the final $y^*$ computed has at most two fractional values and moreover, the expected cost of $y^*$ to the \eqref{LP:basic} is at most $\alpha_q \tilde U'$.

\begin{claim} \label{cl:i1i2}
The solution $y^*$ either is already fully integral, or has exactly two fractional values, say, $y_{i_1}$ and $y_{i_2}$. Moreover it holds that $y^*_{i_1} + y^*_{i_2} = 1$.
\end{claim}
\begin{proof}
WLOG we can assume $y^*(F)=k$, as if $y^*(F) < k$, then we could open more facilities in $y^*$. This along with the fact that $y^*$ has at most two fractional values establishes the proof.
\end{proof}

Now, in case $y^*$ is already integral, \Cref{thm:rounding} already gives a solution of cost $\alpha_q\tilde U'$. Hence, 
we focus on the other case, where, by Claim~\ref{cl:i1i2}, there are exactly two fractional facilities, say $i_1$ and $i_2$. 
We assume that we are given as input the $y^*$ output by~\Cref{alg:iter}, along with the sets $F_j$ for $j \in C$, radius-levels $\ell_j$, and the partition of $C$ into $\cpart$ and $\cfull$ maintained by~\Cref{alg:iter} when it terminated.

Let $C_1 = \set{j \in \cpart \, :\, i_1 \in F_j\ \text{and}\, i_2 \notin F_j }$, and similarly let $C_2 = \set{j \in \cpart \, :\, i_1 \notin F_j \, \text{and}\, i_2 \in F_j }$. By renaming $i_1$ and $i_2$ if necessary, we assume $|C_1| \geq |C_2|$.   Our final solution $\hat y$ will be defined as: $\hat{y}_{i_1} = 1$, $\hat{y}_{i_2} = 0$, and $\hat{y}_i = y^*_i$ for $i \notin \{i_1,i_2\}$. Our algorithm will connect $\cfull \cup C_1$ to their respective nearest facilities in $\hat y$.

\begin{lemma} \label{lem:mcov}
The solution described above connects at least $m$ clients.
\end{lemma}

\begin{proof}
This is easy to see, as it covers all full clients, and the only partial clients which are connected in the original LP solution $y^*$ output by~\Cref{alg:iter} are those in $C_1$ and $C_2$, contributing a total of $y_{i_1} |C_1| + y_{i_2} |C_2|$ to constraint~\eqref{LPiter:coverage}. Since $|C_1| \geq |C_2|$ and $y_{i_1} + y_{i_2} = 1$, we have $|C_1| + |\cfull| \geq y_{i_1} |C_1| + y_{i_2} |C_2| + |\cfull| \geq m$. 
\end{proof}

\begin{lemma} \label{lem:c1cost}
The total connection cost of all clients in $C_1$ is at most $2 \rho \tilde U$.
\end{lemma}

\begin{proof}
Notice that $i_1$ is open in the solution $\hat y$. The facility $i_1$ is not collocated with $S_0$ since there are exactly one integral facility open at each location in $S_0$. The connection cost is thus at most
\begin{align*}
	\sum_{j \in C_1} d^q(i_1, j) \leq \sum_{j: i_1 \in F_j} d^q (i_1, j) \leq 2\rho \tilde U.
\end{align*}
The first inequality holds, since if $j$ was only partially connected to $i_1$, then $i_1 \in F_j$. Also, since $F_j$ can only shrink during the algorithm and by Property~\ref{property:y-star-star}, the second inequality holds.
\end{proof}

It remains to consider what we need to change in Theorem~\ref{thm:rounding} for the analysis of connection cost for full clients, when we change $y^*$ to $\hat y$. Notice that for every $j \in \cfull$, we have one unit of open facility in $\Ball_F(j, \frac{3\tau-1}{\tau-1}D_{\ell_j})$ in $\hat y$. This holds due to Lemma~\ref{lem:5r}, and the fact that $y^*_{i_2}$ is strictly smaller than $1$, meaning that if $i_2$ is in such a ball, then so is $i_1$, or some other facility $i_3$ with $y^*_{i_3} = \hat y_{i_3} = 1$.

Thus, we only need to focus on the set $J = \set{j \in \cfull: i_2 \in B_j}$ of clients and consider the total connection cost of these clients in the new solution. Let $\delta' = \Theta(\delta)$ be a number such that $\delta' \leq \frac{\delta}{4+3\delta}\cdot\frac{(\tau-1)(1-\delta')}{\tau(3\tau-1)}$; setting $\delta' = \frac{\delta}{6}\cdot\frac{\tau-1}{2\tau(3\tau-1)}$ satisfies the property. Let $i^*$ be the nearest open facility to $i_2$ in the solution $\hat y$ and let $t' = d(i_2, i^*)$. Let $J_1 = \set{j \in J:d(j, i_2) \geq \delta' t'}$ and $J_2 = J \setminus J_1 = \set{j \in J: d(j, i_2) < \delta' t'}$.

For every $j \in J_1$, we have $d(j, i^*) \leq d(j, i_2) + d(i_2, i^*) \leq d(j, i_2) + \frac{d(j, i_2)}{\delta'} = (1+1/\delta')d(j, i_2)$. Thus,
\ifdefined\fullversion
	\begin{align*}
		\sum_{j \in J_1} d^q (j, i^*) &\leq \left(1+\frac1{\delta'}\right)^q\sum_{j \in J_1} d^q(j, i_2) \leq \left(1+\frac1{\delta'}\right)^q\cdot 2\rho\tilde U = O\left(\frac\rho{\delta^q}\right) U.
	\end{align*}
\else
	\begin{align*}
		\sum_{j \in J_1} d^q (j, i^*) &\leq \left(1+\frac1{\delta'}\right)^q\sum_{j \in J_1} d^q(j, i_2) \\
		&\leq \left(1+\frac1{\delta'}\right)^q\cdot 2\rho\tilde U = O\left(\frac\rho{\delta^q}\right) U.
	\end{align*}
\fi
To see the second inequality, we know that for every $j \in J_1 \subseteq J$, we have $i_2 \in B_j \subseteq F_j$ at the end of Algorithm~\ref{alg:iter}. Thus, $i_2 \in F_j$ at the beginning, and hence by Property~\ref{property:y-star-star}, we have $\sum_{j \in J_1}d^q(j, i_2) \leq 2\rho \tilde U$.

For every $j \in J_2$, we have $t' - d(i_2, j)$ is at most the distance from $j$ to its the nearest open facility in $\hat y$, which is at most $\frac{3\tau-1}{\tau-1} D_{\ell_j}$. Thus, $R_j \geq \frac{D_{\ell_j}}\tau \geq \frac{\tau-1}{\tau(3\tau-1)}(t' - d(i_2, j)) \geq \frac{\tau-1}{\tau(3\tau - 1)}(1-\delta')t'$.

Let $t = \frac{\tau-1}{\tau(3\tau - 1)}(1-\delta)t'$. Then for every $j \in J_2$, we have $R_j \geq t$. Also, by our choice of $\delta'$, we have $\frac{\delta t}{4+3\delta} \geq \delta't'$. So,
\ifdefined\fullversion
	\begin{align*}
		|J_2| &\leq \left|\set{j \in \Ball_{C}\left(i_2, \frac{\delta t}{4+3\delta}\right):R_j \geq t}\right| \leq \frac{\rho(1+3\delta/4)^q}{(1-\delta/4)^q} \cdot \frac{U}{t^q} = O(\rho)\cdot \frac{U}{t^q}.
	\end{align*}
\else
	\begin{align*}
		|J_2| &\leq \left|\set{j \in \Ball_{C}\left(i_2, \frac{\delta t}{4+3\delta}\right):R_j \geq t}\right| \\
		&\leq \frac{\rho(1+3\delta/4)^q}{(1-\delta/4)^q} \cdot \frac{U}{t^q} = O(\rho)\cdot \frac{U}{t^q}.
	\end{align*}
\fi
The second inequality is by Property~\ref{property:prep-sparse-restated}.  So,
\ifdefined\fullversion
	\begin{align*}
		\sum_{j \in J_2}d^q(j, i^*) &\leq |J_2| ((1+\delta')t')^q \leq O(1)\cdot{t'}^q \cdot O(\rho)\cdot\frac{U}{t^q} = O(\rho) U.
	\end{align*}
\else
	\begin{align*}
		\sum_{j \in J_2}d^q(j, i^*) &\leq |J_2| ((1+\delta')t')^q \\&\leq O(1)\cdot{t'}^q \cdot O(\rho)\cdot\frac{U}{t^q} = O(\rho) U.
	\end{align*}
\fi
Thus, overall, we have $\sum_{j \in J}d^q(j, i^*) \leq O(\rho/\delta^q) U$. This shows that the additional cost incurred due to changing the almost-integral solution to an integral one is at most $O(\rho/\delta^q)U$, finishing the proof of Theorem~\ref{thm:main}.


%% file: matroid.tex
\section{Improved Approximation Algorithm for Matroid Median}
\label{sec:matroid}

The Matroid Median problem~(\matmed) is a generalization of $k$ median where we require the set of open facilities to be an independent set of a given matroid. The $k$ median problem is a special case when the underlying matroid is a uniform matroid of rank $k$.  We now show how we can use our framework to obtain an $\alpha_1$ approximation for \matmed.

Unlike \kmwo, the natural LP relaxation for \matmed only has a small integrality gap, and so we can skip the pre-processing steps and jump directly to the iterative rounding framework. To give a description in the simplest way, we just assume $U=\infty$, $U' = \tilde U'$ is a guessed upper bound on the cost of the optimum solution, $R_j = \infty$ for every $j \in C$ and $m = n$, at the beginning of Section~\ref{sec:framework}. The only change is that $y(F)\leq k$ in \eqref{LP:strong} and \eqref{LP:iter} need to be replaced with the matroid polytope constraints $\sum_{v\in S}y_v \leq \, r_\mathcal{M}(S), \forall S \subseteq F$.  Eventually, all clients $j$ will be moved to $\cfull$ and thus the Constraint~\eqref{LPiter:coverage} becomes redundant. The final $y^*$ is a vertex point in the intersection of a partition matroid polytope and the given matroid polytope. Thus $y^*$ is integral. This leads to an integral solution with expected cost $\alpha_1 U'$, leading to an $\alpha_1$-approximation for \matmed.

%% file: knapsack.tex
\section{Improved approximation algorithm for Knapsack Median} \label{sec:knapsack}

In the knapsack median problem (\knapmed), we are given a knapsack constraint on the facilities instead of a bound on number of facilities opened as in $k$ median.
Formally, for every facility $i$, there is a non-negative weight $w_i$ associated with it, and we have an upper bound $W$ on the total weight of facilities. The problem is to choose a vector $y\in \{0,1\}^n$ that satisfies $\sum_{i\in F}w_iy_i\leq W$ and we minimize the total connection cost of all the clients.
The classical $k$ median is a special case of knapsack median, when all the weights $w_i$ are equal to $1$ and $W$ is equal to $k$. Due to the knapsack constraint, the natural LP for \knapmed has an unbounded integrality gap. We perform the similar preprocessing as for \kmwo (\Cref{sec:preprocessing}), except now we use a simple way to give the upper bound vector $(R_j)$, instead of applying Theorem~\ref{thm:preprocessing}.

\begin{definition}[Sparse Instances for \knapmed]
	\label{def:sparse-kp}
	Suppose we are given an extended \knapmed instance ${\mathcal{I}} = (F, C, d, w, B, S_0)$, and parameters $\rho, \delta \in (0, 1/2)$ and $U \geq 0$. Let $S^*$ be a solution to $\mathcal{I}$ with cost at most $U$, and $(\kappa^*, c^*)$ be the nearest-facility-vector-pair for $S^*$. Then we say the instance ${\mathcal {I}}$ is $(\rho,\delta, U)$-sparse w.r.t the solution $S^*$, if
	\begin{enumerate}[itemsep=0pt,topsep=3pt,label=(\ref{def:sparse-kp}\alph*),leftmargin=*]
		\item \label{property:sparse-star-cheap-kp} for every $i \in S^* \setminus S_0$, we have $\sum_{j \in C:\kappa^*_j = i} c^*_j \leq \rho U$,
		\item \label{property:sparse-sparse-kp} for every $p \in F \cup C$, we have $\big|\Ball_C(p, \delta c^*_p)\big|\cdot c^*_p \leq \rho U$.
	\end{enumerate}
\end{definition}


Using a similar argument as in Theorem~\ref{thm:reduce-to-sparse-instances}, given a \knapmed instance $\mathcal{I}$, we can produce a family of $n^{O(1/\rho)}$ extended \knapmed instances, one of which, say $\mathcal{I'} = (F, C' \subseteq C, d, w, B, \allowbreak S_0 \subseteq S^*)$, is $(\rho, \delta, U)$-sparse w.r.t $S^*$, the optimum solution of $\mathcal{I}$. Moreover, $\frac{1-\delta}{1+\delta}\sum_{j \in C \setminus C'}d(j, S_0) + \sum_{j \in C'} d(j, S^*) \leq U$. Therefore it suffices for us to prove the following theorem:
\begin{theorem}
	\label{thm:main-kp}
	Let $\mathcal{I'} = (F, C', d, w, B, S_0)$ be a $(\rho, \delta, U)$-sparse instance w.r.t some solution $S^*$, for some $\rho, \delta \in (0, 1/2)$ and $U \geq 0$. Let $U'$ be the cost of $S^*$ for $\mathcal{I'}$. There is an efficient randomized algorithm that computes a solution to $\mathcal{I'}$ with expected cost at most $\alpha_1 U' + O(\rho/\delta)U$.
\end{theorem}

Again, we shall replace $\mathcal{I'}$ with $\mathcal{I}$ and $C'$ with $C$.  We use a very simple way to define $R_j$: for every $j$, let $R_j$ be the maximum number $R$ such that $|\Ball_C(j, \delta R)|\cdot R \leq \rho U$. Let $(\kappa^*, c^*)$ be the nearest-facility-vector-pair for $S^*$. By Property~\ref{property:sparse-sparse-kp}, we have, $c^*_j \leq R_j$. We then formulate a stronger LP relaxation:
\begin{equation}
	\min \qquad \sum_{i \in F, j \in C}x_{i, j} d(i, j) \qquad \text{s.t.} \tag{$\text{LP}_{\text{k-strong}}$} \label{LP:strong-kp} 
\end{equation} \vspace*{-20pt}

\ifdefined\fullversion
	\noindent \begin{minipage}{0.35\textwidth}
	\centering
	\begin{alignat}{2}
	        y(F) &\leq k \label{LPs:k-kp} \\
			x_{i, j} &\leq y_i &\qquad &\forall i, j \\
			\sum_{i \in F}x_{i, j} &= 1 &\qquad &\forall j\\
	        y_i &= 1 &\qquad & \forall i \in S_0 \label{LPs:S0-kp}
	\end{alignat}
	\end{minipage}
	\begin{minipage}{0.65\textwidth}
	\centering
	\begin{alignat}{2}
	        x_{i,j} &= 0 &\qquad &\forall i, j \, \text { s.t } \, d(i,j) >  R_j \label{LPs:R-kp} \\
	        x_{i,j} &= 0 &\qquad &\forall i \in F\setminus S_0, j \, \text { s.t } \, d(i,j) >  \rho U \label{LPs:R-d-kp}\\
	        \sum_{i} d(i,j) x_{i,j} & \leq \rho U y_i &\qquad & \forall i \in F \setminus S_0 \label{LPs:star-kp} \\[5pt]
	        \nonumber
	\end{alignat}
	\end{minipage}\medskip
\else
\begin{align}
        \sum_i w_i y_i &\leq W \label{LPs:k-kp} \\
		x_{i, j} &\leq y_i &\qquad &\forall i, j \\
		\sum_{i \in F}x_{i, j} &= 1 &\qquad &\forall j\\
        y_i &= 1 &\qquad & \forall i \in S_0 \label{LPs:S0-kp}\\
        x_{i,j} &= 0 &\qquad &\forall i, j \, \text { s.t } \, d(i,j) >  R_j \label{LPs:R-kp} \\
        x_{i,j} &= 0 &\qquad &\forall i \in F\setminus S_0, j \, \text { s.t } \, d(i,j) >  \rho U \label{LPs:R-d-kp}\\
        \sum_{i} d(i,j) x_{i,j} & \leq \rho U y_i &\qquad & \forall i \in F \setminus S_0 \label{LPs:star-kp}
\end{align}
\fi

There is a solution $(x, y)$ to the above LP with cost at most $U'$, as the indicator vector $(x, y)$ for solution $S^*$ satisfies all the above constraints.  After this LP, we eliminate the $x_{i,j}$ variables akin to~\Cref{subsec:splitting}, then perform a randomized discretization of distances identical to the one in~\Cref{sec:discretize}, and obtain a feasible solution to the following auxiliary LP.  
Like in~\Cref{sec:framework}, we begin with $\cpart = C$, $\cfull = \cstar = \emptyset$.

Then we formulate an LP for the iterative rounding. Now we can require $y(F_j)  = 1$ for every $j \in \cpart$ and thus the coverage constraint is not needed:
\begin{multline}
  \min \qquad \sum_{j \in \cpart} \sum_{i \in F_j} {d'}(i, j)y_i  \ifdefined\fullversion\else\\\fi
  +  \sum_{j \in \cfull}\left( \sum_{i \in B_j} {d'}(i, j)y_i + (1 - y(B_j)) d_{\ell_j} \right)  \qquad \textrm {s.t.}  \tag{$\text{LP}_{\text{k-iter}}$} \label{knapsack-LP:iter}
\end{multline}
\ifdefined\fullversion
	\vspace*{-15pt}
	
	\noindent\begin{minipage}{0.5\linewidth}
		\begin{alignat}{2}
	  \sum_{i} w_i y_i &\leq \, W&\qquad  & \label{knapsack-LPiter:1} \\
			y(F_j) &= \, 1 &\qquad &\forall j \in \cstar \cup \cpart \label{knapsack-LPiter:y-full}
		\end{alignat}
	\end{minipage}
	\begin{minipage}{0.5\linewidth}
		\begin{alignat}{2}
			y(B_j) &\leq \, 1 &\qquad &\forall j \in \cfull  \label{knapsack-LPiter:y-B}\\[10pt]
	   		y_i &\in  [0, 1] &\qquad &\forall i \in F \label{knapsack-LPiter:6}
		\end{alignat}
	\end{minipage}\medskip
\else
	\vspace*{-20pt}

		\begin{alignat}{2}
	  \sum_{i} w_i y_i &\leq \, W&\qquad  & \label{knapsack-LPiter:1} \\
			y(F_j) &= \, 1 &\qquad &\forall j \in \cstar \cup \cpart \label{knapsack-LPiter:y-full}\\
			y(B_j) &\leq \, 1 &\qquad &\forall j \in \cfull  \label{knapsack-LPiter:y-B}\\
	   		y_i &\in  [0, 1] &\qquad &\forall i \in F \label{knapsack-LPiter:6}
		\end{alignat}
\fi
We then run the same~\Cref{alg:iter} except we solve~\eqref{knapsack-LP:iter} each time instead of~\eqref{LP:iter}. 
Eventually every client will become a full client, i.e., $\cfull = C$ and $\cpart = \emptyset$. 
At the end of iterative rounding algorithm, the only tight constraints are \eqref{knapsack-LPiter:1}, ~\eqref{knapsack-LPiter:y-full} and \eqref{knapsack-LPiter:6}. Now, we use the fact that the intersection of a laminar family of constraints and a knapsack polytope is an almost-integral polytope with at most two fractional variables. Moreover, following the same analysis as the proof of~\Cref{thm:rounding}, we can bound the expected connection cost to be at most $\frac{3\tau - 1}{\ln {\tau}} U' = \alpha_1 U'$, where the expectation is over the random choice for $a$ in discretization step.

In order to get an integral solution $\hat y$, we let $\hat y = y$ and consider following two cases. If there is exactly one strictly fractional facility, we call it $i_2$ and close it in the solution $\hat y$. If there are two strictly fractional facilities $i_1$ and $i_2$, then we must have $y^*_{i_1} + y^*_{i_2} = 1$. Let $i_1$ be the facility with smaller weight and $i_2$ be the one with bigger weight. We open $i_1$ and close $i_2$ in $\hat y$. Thus, we ensure that the knapsack constraint is not violated in the integral solution.

Unlike \kmwo, there is only one type of cost that we incur during this process since all clients are fully connected. The clients that were using $i_2$ partially now need to be connected to their nearest open facility.  We bound the increase of connection cost of these clients.  Again, it suffices to focus on the set $J = \set{j \in \cfull: i_2 \in B_j}$ of clients and consider the total connection cost of these clients in the new solution $\hat y$. In the solution $\hat y$, for every client $j \in \cfull = C$, there is always a open facility with distance $D_{\ell_j}$ from $j$.

Let $\delta' = \Theta(\delta)$ be a number such that $\delta' \leq \frac{(\tau-1)(1-\delta')}{2\tau(3\tau-1)}\delta$; setting $\delta' = \frac{\tau-1}{4\tau(3\tau-1)}\delta$ satisfies the property. Let $i^*$ be the nearest open facility to $i_2$ in the solution $\hat y$ and let $t = d(i_2, i^*)$. Let $J_1 = \set{j \in J:d(j, i_2) \geq \delta' t}$ and $J_2 = J \setminus J_1 = \set{j \in J: d(j, i_2) < \delta' t}$.

For every $j \in J_1$, we have $d(j, i^*) \leq d(j, i_2) + d(i_2, i^*) \leq d(j, i_2) + \frac{d(j, i_2)}{\delta'} = (1+1/\delta')d(j, i_2)$. Thus,
\begin{align*}
	\sum_{j \in J_1} d (j, i^*) \leq \left(1+\frac1{\delta'}\right)\sum_{j \in J_1} d(j, i_2) \leq \left(1+\frac1{\delta'}\right)\cdot 2\rho U = O\left(\frac\rho{\delta}\right) U.
\end{align*}
To see the second inequality, we know that for every $j \in J_1 \subseteq J$, we have $i_2 \in B_j \subseteq F_j$ in the end of Algorithm~\ref{alg:iter}. Thus, $i_2 \in F_j$ after the splitting. By the property that each star has small cost, we have $\sum_{j \in J_1}d(j, j_2) \leq 2\rho U$.


Assuming $J_2 \neq \emptyset$. Fix any ${j^*} \in J_2$. $t - \delta' t \leq t - d({j^*}, i_2) = d(i_2, i^*) - d({j^*}, i_2) \leq d({j^*}, i^*) \leq \frac{3\tau-1}{\tau-1} D_{\ell_j} \leq \frac{(3\tau-1)\tau}{\tau-1}R_j$. Thus $R_j\geq \frac{(1-\delta')(\tau-1)}{\tau(3\tau-1)} t$.  By our choice of $\delta'$, we have $\delta R_j \geq 2\delta' t$. So,
\begin{align*}
	|J_2| \leq \left|\Ball_{C}\left({j^*}, \delta R_j\right)\right| \leq \frac{\rho U}{R_j}.
\end{align*}
So,
\begin{align*}
	\sum_{j \in J_2}d(j, i^*) \leq |J_2| (1+\delta')t \leq \frac{\rho U}{R_j} \cdot O(1) R_j = O(\rho) U.
\end{align*}

Thus, overall, we have $\sum_{j \in J}d(j, i^*) \leq O(\rho/\delta) U$. This shows that the additional cost incurred due to changing the almost-integral solution to an integral one is at most $O(\rho/\delta)U$, finishing the proof of Theorem~\ref{thm:main-kp}.